\documentclass[journal,onecolumn]{IEEEtran}

\newcommand{\LF}{\ensuremath{\mathit{LF}}}
\newcommand{\BWT}{\ensuremath{\mathit{BWT}}}

\newcommand{\ISA}{\ensuremath{\mathit{ISA}}}
\newcommand{\SA}{\ensuremath{\mathit{SA}}}

\newcommand{\LCP}{\ensuremath{\mathit{LCP}}}

\newcommand{\rank}{\ensuremath{\mathrm{rank}}}

\newcommand{\dd}{\mathinner{.\,.}}

\usepackage{graphicx}
\usepackage{color}
\usepackage{amsmath}
\usepackage{amssymb}
\usepackage{amsthm}
\usepackage{todonotes}

\newtheorem{definition}{Definition}
\newtheorem{theorem}{Theorem}
\newtheorem{lemma}[theorem]{Lemma}
\newenvironment{example}{{\medskip \noindent \em Example:\ }}{}

\begin{document}

\title{On the Approximation Ratio of Ordered Parsings}
\author{%
Gonzalo Navarro, Carlos Ochoa, and Nicola Prezza%
\thanks{An early version of this paper appeared in {\em Proc. LATIN'18} \cite{GNP18}.}
\thanks{Gonzalo Navarro and Carlos Ochoa are with
Center for Biotechnology and Bioengineering (CeBiB), 
Department of Computer Science, University of Chile, Chile.
Gonzalo Navarro is also with Millennium Institute for Foundational Research on Data
(IMFD), Chile. Nicola Prezza is with 
Department of Computer Science, University of Pisa, Italy.}
}%

\maketitle

\begin{abstract}
Shannon's entropy is a clear lower bound for statistical compression. The
situation is not so well understood for dictionary-based compression. A 
plausible lower bound is $b$, the least number of phrases of a general 
bidirectional parse of a text, where phrases can be copied from anywhere 
else in the text. Since computing $b$ is NP-complete, a popular gold standard
is $z$, the number of phrases in the Lempel-Ziv parse of the text, which is
the optimal one when phrases can be copied only from the left. While $z$ can
be computed in linear time with a greedy algorithm,
almost nothing has been known for decades about its approximation 
ratio with respect to $b$. In this paper we prove that $z=O(b\log(n/b))$, where
$n$ is the text length. We also show that the bound is tight as a function of
$n$, by exhibiting a text family where $z = \Omega(b\log n)$. 
Our upper bound is obtained by building a run-length context-free grammar based on a locally
consistent parsing of the text. Our lower bound is obtained by relating $b$ 
with $r$, the number of equal-letter runs in the Burrows-Wheeler transform of 
the text.
We proceed by observing that Lempel-Ziv is just one particular case of
{\em greedy} parses, meaning that the optimal value of $z$ is obtained by 
scanning the text and maximizing the phrase length at each step, and of 
{\em ordered} parses, meaning that there is an increasing order between phrases
and their sources. As a new example of 
ordered greedy parses, we introduce {\em lexicographical} parses, where phrases can 
only be copied from lexicographically smaller text locations. We prove 
that the size $v$ of the optimal lexicographical parse is also obtained greedily
in $O(n)$ time, that $v=O(b\log(n/b))$, and that there exists a text family 
where $v = \Omega(b\log n)$. Interestingly, we also show that $v = O(r)$ 
because $r$ also induces a lexicographical parse, 
whereas $z = \Omega(r\log n)$ holds on some text families. 
We obtain some results on parsing complexity and size that hold on some general
classes of greedy ordered parses.
In our way, we also prove other relevant bounds between compressibility measures, 
especially with those related to smallest grammars of various types generating (only) the text.
\end{abstract}



\section{Introduction}

Shannon \cite{Sha48} defined a measure of entropy that serves as a lower bound 
to the attainable compression ratio on any source that emits symbols according
to a certain probabilistic model. An attempt to measure the compressibility of
finite texts $T[1..n]$, other than the non-computable Kolmogorov complexity
\cite{Kol65}, is the notion of empirical entropy \cite{CT06}, where some
probabilistic model is assumed and its parameters are estimated from the 
frequencies observed in the text.
Other measures that, if the text is generated from a probabilistic source, 
converge to its Shannon entropy, are derived from the Lempel-Ziv parsing 
\cite{LZ76} or the grammar-compression \cite{KY00} of the text.

Some text families, however, are not well modeled as coming from a 
probabilistic source. A very recent case is that of {\em highly repetitive
texts}, where most of the text can be obtained by copying long blocks from 
elsewhere in the
same text. Huge highly repetitive text collections are arising from the 
sequencing of myriads of genomes of the same species, from versioned document
repositories like Wikipedia, from source code repositories like GitHub,
etc. Their growth is outpacing Moore's Law by a wide margin \cite{Plos15}.
Understanding the compressibility of highly repetitive texts is important to
properly compress those huge collections.

Lempel-Ziv and grammar compression are particular cases of so-called {\em
dictionary techniques}, where a set of strings is defined and the text is
parsed as a concatenation of those strings. On repetitive collections, the empirical
entropy ceases to be a relevant compressibility measure; for example the
$k$th order per-symbol entropy of $TT$ is the same as that of $T$,
if $k \ll n$ \cite[Lem.~2.6]{KN13}, whereas this entropy measure is generally
meaningless for $k > \log n$ \cite{Gag06}. Some dictionary 
measures, instead, capture much better the compressibility of repetitive texts.
For example, while an individual genome can rarely be compressed to much less
than 2 bits per symbol, Lempel-Ziv has been reported to compress collections of 
human genomes to less than 1\% \cite{FLCB11}. Similar compression ratios are 
reported in Wikipedia.%
\footnote{{\tt https://en.wikipedia.org/wiki/Wikipedia:Size\_of\_Wikipedia}}

Despite the obvious practical relevance of these compression methods,
there is not a clear entropy measure useful for highly repetitive texts. The 
number $z$ of phrases generated by the Lempel-Ziv parse \cite{LZ76} is often 
used as a gold standard, possibly because it can be implemented in linear time
\cite{RPE81} and is never larger than $g$, the size of the smallest 
context-free grammar that generates the text \cite{Ryt03,CLLPPSS05}. However, 
$z$ is not
so satisfactory as an entropy measure: the value may change if we reverse the 
text, for example. 
A much more robust lower bound on compressibility is $b$,
the size of the smallest {\em bidirectional (macro) scheme} \cite{SS82}. Such
a scheme parses the text into phrases such that each phrase appears somewhere
else in the text (or it is a single explicit symbol), in a way that makes it 
possible
to recover the text by copying source to target positions in an appropriate
order. This is arguably the strongest possible dictionary method, but finding
the smallest bidirectional scheme is NP-complete \cite{Gal82}. A relevant
question is then how good is the Lempel-Ziv parse as an efficiently 
implementable approximation to the smallest bidirectional scheme.
Almost nothing is known in this respect, except that there are string 
families where $z$ is nearly $2b$ \cite{SS82}.\footnote{An article implying
$z=\Omega(b\log n)$ \cite[corollary in 3rd page]{HC97} has a mistake: their string $D$ is also parsed
in $\Theta(N)$ phrases by LZ76, not $\Theta(N\log N)$.}

In this paper we finally give a tight approximation ratio for $z$, showing that
the gap is larger than what was previously known. We prove that 
$z = O(b\log(n/b))$, and that this bound is tight as a function of $n$, by
exhibiting a string family where $z = \Omega(b\log n)$. To prove the upper 
bound, we show how to build a run-length context-free grammar 
\cite{NIIBT16} (i.e., allowing constant-size rules of the form $X \rightarrow Y^t$) 
of size $g_{rl} = O(b\log(n/b))$.
This is done by carrying out several rounds of locally consistent parsing 
\cite{Jez15} on top of $T$, reducing the resulting blocks to nonterminals in
each round, and showing that new nonterminals appear only in the boundaries
of the phrases of the bidirectional scheme. We then further prove that $z \le
2g_{rl}$, by extending a classical proof \cite{CLLPPSS05} that relates grammar
with Lempel-Ziv compression. To prove the lower bound, we consider another
plausible compressibility measure: the number $r$ of equal-symbol runs in the 
Burrows-Wheeler transform (BWT) of the text \cite{BW94}. We prove that the BWT 
induces a valid bidirectional scheme, and thus $r = \Omega(b)$. Then the bound
follows from the family of Fibonacci words, where $z = \Theta(\log n)$ and $r=O(1)$. 
The latter result, however, assumes that lexicographical comparisons regard the strings
as cyclic, instead of the more natural notion we use here. We then study the Fibonacci 
words under our model, to show that $r=O(1)$ still holds in the even members of the family.

We then show that Lempel-Ziv is just one valid example of interesting parses holding
(i) that they impose an increasing order between sources and targets, and (ii) can be 
efficiently computed with a greedy algorithm. We define a weak and a strong notion of 
order, which coincide in the case of the text-precedence order used by Lempel-Ziv.
We design a greedy polynomial-time algorithm that always finds the optimum parse that 
strongly satisfies a given order. We also prove that the optimum parse weakly satisfying
a given order is of size $O(g)$, and even $O(g_{rl}) \subseteq O(b\log(n/b))$ if sources can overlap targets.

We then give a concrete parsing arising from our generalization. We define $v$, the
size of the optimal {\em lexicographic parse} of the text, where each phrase
must point to a lexicographically smaller one (both seen as text suffixes). In such
an order, the weak and strong versions are also equivalent. Thus, it holds that
$v = O(g_{rl}) \subseteq O(b\log(n/b))$. Further, we show that $v$ can be computed in 
linear time, with a very practical algorithm. We also show that $r$ induces a 
lexicographical parse, thus $v = O(r)$. Since, instead, $z$ can be $\Omega(r\log n)$,
our new greedy parse asymptotically improves the Lempel-Ziv parse on some string families.
We also show that $b=O(1)$ and $v=\Theta(\log n)$ (and thus $v=\Omega(b\log n)$) on the
odd Fibonacci words, but we have not found a family where $z = o(v)$. We show that $v$ and $z$ perform comparably on a set of benchmark repetitive texts.

Finally, we consider the size $c$ of the smallest {\em collage system} \cite{KMSTSA03}, which
adds to run-length context-free grammars the power to truncate a prefix or a suffix of a nonterminal. Little was known about this measure, except that $c = O(\min\{g_{rl},z\log z\})$. By extending the ideas of the article, we prove that $b = O(c)$, $c = O(z)$, and that there exists string families where $c = \Omega(b\log n)$, for a subclass we call {\em internal collage systems} where all the productions appear in $T$.


\section{Basic Concepts} \label{sec:basics}

\begin{table}[t]
\caption{Notation assumed all along the paper, including theorems.}
\begin{center}
\begin{tabular}{l|l}
Symbol & Meaning \\
\hline
$T$      & Text to be parsed or compressed \\
$n$      & Text length \\
$\sigma$ & Text alphabet size \\
$f$      & Target-to-source mapping in a parsing of $T$ \\
$H_k$    & Per-symbol $k$th order empirical entropy of $T$ \\
\hline
$b$      & Size of smallest bidirectional scheme for $T$ \\
$z$      & Size of Lempel-Ziv parse for $T$ \\
$z_{no}$ & Size of Lempel-Ziv parse for $T$ not allowing overlaps \\
$g$      & Smallest size (number of rules) of an SLP generating $T$ \\
$g_{rl}$ & Smallest size (number of rules) of an RLSLP generating $T$ \\
$c$      & Smallest size (number of rules) of an internal collage system generating $T$\\
$r$      & Number of runs in the BWT of $T$ \\
$u$      & Smallest size of an ordered parse for $T$ \\
$v$      & Size of the lex-parse for $T$ \\
\hline
$f_k$    & $k$th Fibonacci number \\
$F_k$    & $k$th Fibonacci word \\
\hline
\end{tabular}
\end{center}
\label{tab:notation}
\end{table}

We review basic concepts about strings, compression measures, and others. 
Table~\ref{tab:notation} summarizes our notation along the article.

\subsection{Strings and String Families}

A {\em string} (or {\em word}) is a sequence $S[1\dd\ell] = S[1] S[2] \cdots S[\ell]$ of symbols, of length $|S|=\ell$. 
A {\em substring} (or {\em factor}) $S[i] \cdots S[j]$ of $S$ is denoted $S[i\dd j]$. A {\em suffix} of $S$ is 
a substring of the form $S[i\dd\ell] = S[i\dd]$, and a {\em prefix} is a substring of the form $S[1\dd i] = S[\dd i]$. 
The juxtaposition of strings and/or symbols represents their concatenation; the explicit infix
 operator ``$\cdot$'' can also be used.

We will consider parsing or compressing a string $T[1\dd n]$, called the 
{\em text}, over alphabet $[1\dd\sigma]$.
We assume that $T$ is terminated by the special symbol $T[n] = \$$, which is 
lexicographically smaller than all the others. This makes any lexicographic 
comparison between suffixes well-defined: when a suffix is a prefix of another, the prefix is lexicographically smaller.

We use various infinite families of strings along the article, to prove lower and upper bounds. An important family we use are the {\em Fibonacci words}, defined as follows.

\begin{definition} \label{def:fibonacci}
The Fibonacci word family is defined as $F_1=b$, $F_2=a$, and for all 
$k>2$, $F_k = F_{k-1} \cdot F_{k-2}$. The length of $F_k$ is $f_k$, the
$k$th Fibonacci number, defined as $f_1=f_2=1$ and, for $k>2$, $f_k = 
f_{k-1}+f_{k-2}$.
\end{definition}

To obtain results compatible with the usual convention that a prefix of a suffix is
lexicographically smaller than it, we will use a variant of the family that has
the terminator \$ (virtually) appended.  

Another family we will use is the {\em de Bruijn sequence} of order $k$ on an
alphabet of size $\sigma$. It contains all the distinct substrings of length 
$k$ over $[1\dd \sigma]$, and it is of the minimum possible length,
$\sigma^k+\sigma-1$.

\subsection{Bidirectional Schemes $(b)$} \label{sec:bidir}

A {\em bidirectional scheme} \cite{SS82} partitions $T[1\dd n]$ into $b$ {\em phrases} 
$B_1, \ldots, B_b$, such that each phrase $B_i = T[t_i\dd t_i+\ell_i-1]$ is either 
(1) copied from another substring $T[s_i\dd s_i+\ell_i-1]$ 
(called the phrase {\em source}) with $s_i \not= t_i$ and $\ell_i \ge 1$, which may overlap 
$T[t_i\dd t_i+\ell_i-1]$, or (2) formed by $\ell_i=1$ explicit symbol. The
phrases of type (1) are also called {\em targets} of the copies.
The bidirectional scheme is {\em valid} if there is an order in which the
sources $s_i+j$ can be copied onto the targets $t_i+j$ so that all the positions
of $T$ can be inferred.

A bidirectional scheme implicitly defines a function $f:[1\dd n] \rightarrow [1\dd n] \cup \{-1\}$ so that, 
  \[
    \begin{cases} 
      f(t_i+j) = s_i+j, & \textrm{if~} T[t_i\dd t_i+\ell_i-1] \textrm{~is copied from~} T[s_i\dd s_i+\ell_i-1] \textrm{~and~} 0\le j<\ell_i \textrm{~(case 1),} \\
      f(t_i) = -1, & \textrm{if~} T[t_i] \textrm{~is an explicit symbol (case 2).}
    \end{cases}
  \]

Being a valid scheme is equivalent to saying that $f$ has no cycles, that is,
there is no $k>0$ and $p$ such that $f^k(p)=p$. Which is the same, for each
$p$ there exists $k \ge 0$ such that $f^k(p) = -1$. We can then recover each
non-explicit text position $p$ from the explicit symbol $T[f^{k-1}(p)]$.

We use $b$ to denote the smallest bidirectional scheme, which is NP-complete
to compute \cite{Gal82}.

\begin{example} 
Consider the text $T=alabaralalabarda\$$. A bidirectional scheme
of $b=10$ phrases is $ala|\underline{b}|a|\underline{r}|\underline{a}|\underline{l}|alabar|\underline{d}|a|\underline{\$}$, where we have underlined the explicit symbols. For example, the source of phrase $B_1 = T[1\dd 3] = ala$ is $T[7\dd 9]$, and the source of phrase $B_7 = T[9\dd 14] = alabar$ is $T[1\dd 6]$. To extract $T[11]$, we follow the chain $f(11)=3$, $f(3)=9$,
$f(9) = 1$, $f(1)=7$, and $f(7)=-1$ because it is an explicit symbol. We then learn
that $T[11]=T[3]=T[9]=T[1]=T[7]=a$.
\end{example}

\subsection{Lempel-Ziv Parsing $(z,z_{no})$}
\label{sec:lz}

Lempel and Ziv \cite{LZ76} define a parsing of $T$ into the fewest possible
phrases $T = Z_1 \cdots Z_z$, so that each phrase $Z_i$ occurs as a substring (but 
not a suffix) of $Z_1 \cdots Z_i$, or an explicit symbol. This means that
the source $T[s_i\dd s_i+\ell_i-1]$ of the target $Z_i = T[t_i\dd t_i+\ell_i-1]$
must satisfy $s_i < t_i$, but sources and targets may overlap. A parsing where sources precede targets in $T$ is called {\em left-to-right}. It turns out
that the {\em greedy} left-to-right parsing, which creates the phrases from
$Z_1$ to $Z_z$ and at each step $i$ maximizes $\ell_i$ (and inserts an
explicit symbol if $\ell_i=0$), indeed produces the 
optimal number $z$ of phrases \cite[Thm.~1]{LZ76}. Further, the parsing can 
be obtained in $O(n)$ time \cite{RPE81,SS82}. This is what we call the {\em Lempel-Ziv parse} of $T$.

If we disallow that a phrase overlaps its source, that is, $Z_i$ must be a 
substring of $Z_1 \cdots Z_{i-1}$ or a single symbol, then we call $z_{no}$ 
the number of phrases obtained. In this case it is also true that the
greedy left-to-right parsing produces the optimal number $z_{no}$ of phrases
\cite[Thm.~10 with $p=1$]{SS82}. Since the Lempel-Ziv parsing allowing overlaps
is optimal among all left-to-right parses, we also have that $z_{no} \ge z$.
This parsing can also be computed in $O(n)$ time \cite{CIKRW12}.
Note that, on a text family like $T=a^n$, it holds that $z=2$ and
$z_{no} = \Theta(\log n)$, and thus it holds that $z_{no} = \Omega(z\log n)$.

Little is known about the relation between $b$ and $z$ except that $z \ge b$
by definition ($z$ is the smallest {\em left-to-right} parsing) and that, for
any constant $\epsilon>0$, there is an infinite family of strings for which
$b < (\frac{1}{2}+\epsilon)\cdot\min(z,z^R)$ \cite[Cor.~7.1]{SS82}, where
$z^R$ is the $z$ value of the reversed string.

Apart from being used as a gold standard to measure repetitiveness, the size
of the Lempel-Ziv parse is bounded by the statistical entropy \cite{LZ76}. In 
particular, if $H_k$ denotes the per-symbol $k$-th order empirical entropy of 
the text \cite{Man01}, then it 
holds that $z_{no} \log_2 n \le nH_k + o(n\log_\sigma n)$ whenever 
$k=o(\log_\sigma n)$ \cite{KM00} (thus, in particular, $z_{no} =
O(n/\log_\sigma n)$).

\begin{example}
Consider again the text $T=alabaralalabarda\$$. The Lempel-Ziv parse (with or withour overlaps) has $z=z_{no}=11$ phrases,
$\underline{a}|\underline{l}|a|\underline{b}|a|\underline{r}|ala|labar|\underline{d}|a|\underline{\$}$, where we have underlined the explicit symbols. For example, the source of phrase $Z_7 = T[7\dd 9] = ala$ is $T[1\dd 3]$, and the source of phrase $B_8 = T[10\dd 14]$ is $T[2\dd 6]$.
\end{example}

\subsection{Grammar Compression $(g,g_{rl})$} \label{sec:gram}

Consider a context-free grammar that generates $T$ and only $T$ \cite{KY00}. 
For simplicity we stick to the particular case of {\em straight-line programs (SLPs)}, which are sequences of
rules of the form $A \rightarrow a$ and $A \rightarrow BC$, where $a$ is a 
terminal and $A,B,C$ are nonterminals. Each nonterminal is defined as the 
left-hand side of exactly one rule, and the right-hand nonterminals must have 
been defined before in the sequence. The {\em expansion} of each nonterminal is the string it generates, that is, $exp(A)=a$ if $A \rightarrow a$ and $exp(A) = exp(B) \cdot exp(C)$ if $A \rightarrow BC$. The {\em initial symbol} of the SLP is the last nonterminal $S$ in the sequence, so that the SLP generates the text $T=exp(S)$. The {\em size} of the SLP is its number of rules; it is assumed that every nonterminal is reachable from the initial symbol. We can stick
to SLPs to obtain asymptotic results because any context-free grammar with
size $g$ (sum of lengths of right-hands of rules) can be easily converted into
an SLP of size $O(g)$.
In general, we will use $g$ to denote the minimum possible size of an SLP
that generates $T$, which is NP-complete to compute \cite{Ryt03,CLLPPSS05}. 

If we allow, in addition, rules of the form $X \rightarrow Y^t$ for an integer
$t > 0$, the result is a {\em run-length SLP (RLSLP)} \cite{NIIBT16}. The rule, assumed
to be of
size 2, means that $X$ expands to $t$ copies $Y$, $exp(X) = exp(Y)^t$.
We will use $g_{rl}$ to denote the size of the smallest RLSLP that generates
$T$, that is, its number of rules. Thus, it is clear that $g_{rl} \le g$. Further, on the string family 
$T=a^n$ it holds that $g_{rl}=2$ and $g=\Theta(\log n)$, and thus it holds that
$g = \Omega(g_{rl}\log n)$ (as well as $z_{no} = \Omega(g_{rl}\log n))$.

A well-known relation between $z_{no}$ and $g$ is $z_{no} \le g = 
O(z_{no}\log(n/z_{no}))$ \cite{Ryt03,CLLPPSS05}. Further, it is known that
$g = O(z\log(n/z))$ \cite[Lem.~8]{Gaw11}. Those papers exhibit
$O(\log n)$-approximations to the smallest grammar, as well as several
others \cite{Sak05,Jez15,Jez16}. A negative result about the approximation are
string families where $g = \Omega(z_{no}\log n/\log\log n)$ 
\cite{CLLPPSS05,HLR16} 
and even $g_{rl} = \Omega(z_{no}\log n/\log\log n)$ \cite{BGLP17}.
The size $g$ is also bounded in terms of the statistical entropy \cite{KY00}
and of the empirical entropy \cite{ON18}, thus it always holds
$g = O(n/\log_\sigma n)$.

The {\em parse tree} of an SLP has a root labeled with the initial symbol and leaves labeled with terminals, which spell out $T$ when read left to right. Each internal node labeled $A$ has a single leaf child labeled $a$ if $A \rightarrow a$, or two internal children labeled $B$ and $C$ if $A \rightarrow BC$.
The {\em grammar tree} prunes
the parse tree by leaving only one internal node labeled $X$ for each
nonterminal $X$; all the others are pruned and converted to leaves. Then, for
an SLP of size $g$, the grammar tree has exactly $g$ internal nodes. 
Since the right-hand sides of the rules are of size 1 or 2, each internal 
node has 1 or 2 children, and thus the total number of nodes is at most $2g+1$.
Therefore, the grammar tree has at most $g+1$ leaves.

\begin{example}
We can generate the text $T=alabaralalabarda\$$ with an SLP
of $g=16$ rules: $A \rightarrow a$, $B \rightarrow b$, $D \rightarrow d$, $L \rightarrow l$, 
$R \rightarrow r$, $Z \rightarrow \$$,
$C \rightarrow AL$, $E \rightarrow AB$, $F \rightarrow AR$, $G \rightarrow DA$, 
$H \rightarrow CE$, $I \rightarrow HF$, $J \rightarrow IC$, $K \rightarrow IG$,
$M \rightarrow JK$, $N \rightarrow MZ$. The nonterminal $N$ is the initial symbol.
Figure~\ref{fig:grammar} illustrates the parse and the grammar trees.
\end{example}

\begin{figure}[t]
\begin{center}
\includegraphics[width=0.7\textwidth]{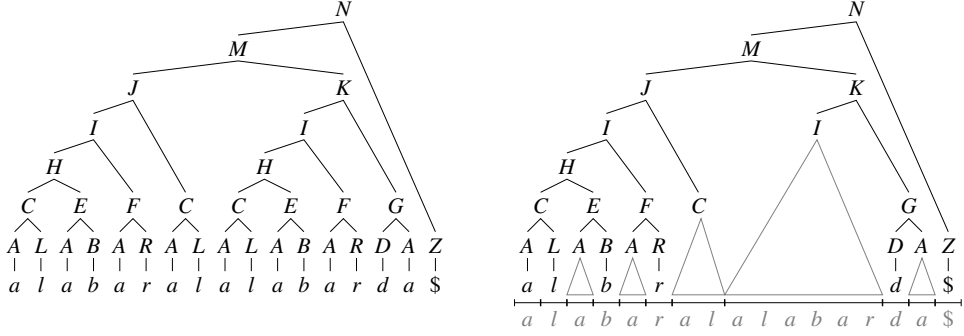}
\end{center}
\caption{The parse tree (left) and the grammar tree (right) of an example text. Only the black elements on the right form the grammar tree; the text coverage is conceptual.}
\label{fig:grammar}
\end{figure}

\subsection{Collage Systems ($c$)} \label{sec:collage}

{\em Collage systems} \cite{KMSTSA03} are a generalization of RLSLPs that also
support {\em truncation}. Specifically, nonterminals can be of the form
$A \rightarrow a$ for a terminal $a$, $A \rightarrow BC$ for previous
nonterminals $B$ and $C$, and $A \rightarrow B^k$, $A \rightarrow B^{[t]}$
and $A \rightarrow ^{[t]}\!\!\!B$ for a previous nonterminal $B$ and positive
integers $k$ and $t$. The last two rules mean that $exp(A) = exp(B)[1\dd t]$ and $exp(A) = exp(B)[|exp(B)|-t+1\dd |exp(B)|]$, respectively (it must hold that $t \le |exp(B)|)$. We denote by $c$ the number of rules
of the smallest collage system generating (only) a text $T$.

Few relations are known between $c$ and other repetitiveness measures, other
than $c \le g_{rl}$ and $c = O(z\log z)$ \cite{KMSTSA03}.

\begin{example}
The following collage system to generate the text $T=alabaralalabarda\$$
is actually less efficient than the SLP (it uses $c=17$ rules), but it illustrates all the operations:
$A \rightarrow a$, $B \rightarrow b$, $D \rightarrow d$, $L \rightarrow l$, $R \rightarrow r$, 
$Z \rightarrow \$$, $C \rightarrow AL$, $E \rightarrow C^3$, $F \rightarrow BA$, $G \rightarrow FR$, 
$H \rightarrow DA$, $I \rightarrow HZ$, $J \rightarrow E^{[5]}$, $K \rightarrow JG$,
$M \rightarrow \,^{[6]} \!K$, $N \rightarrow MK$, $O \rightarrow NI$. The nonterminal $O$ is the initial symbol. This example also illustrates that the concept of parse and grammar tree do 
not apply to collage systems; for example the nonterminal $E$ expands to $alalal$, which does not
exist in the text.
\end{example}\medskip

In this article we will be 
interested in a subclass we call {\em internal collage systems}, where there is a path of non-truncation rules from the initial symbol to every nonterminal. This implies that, every time we use a truncation rule on a nonterminal $A$, the whole $exp(A)$ appears somewhere else in $T$. Since it is not obvious that, in internal collage systems, we can use a prefix plus a suffix truncation to extract a substring of another rule, we explicitly allow in internal collage systems for {\em substring truncation} rules $A \rightarrow B^{[t,t']}$, with $1\le t\le t' \le |exp(B)|$, meaning that $exp(A) = exp(B)[t\dd t']$.

Note that any upper bound we prove for the size $c$ of the smallest internal collage system also holds for the smallest general collage system. In particular, we prove $c=O(z)$, which is an improvement upon the previous result $c=O(z\log z)$ that holds for the smallest general collage system \cite{KMSTSA03}. Instead, an existing lower bound on $c$ of the form  $\gamma=O(c)$, where $\gamma$ is the size of the smallest ``attractor'' for $T$ \cite[Thm.~3.5]{KP18}, holds in fact only for internal collage systems, because it assumes, precisely, that the expansion of every nonterminal appears in $T$.\footnote{For example, with the collage system $A \rightarrow a$, $B \rightarrow b$, $A' \rightarrow A^5$, $B' \rightarrow B^5$, $C \rightarrow AB$, $D \rightarrow C^{[9]}$, and the
initial symbol $E \rightarrow ^{[8]} \!\! D$, we generate the text $T=a^4 b^4$. However, because $C$ does not appear in $T$, they fail to place an attractor element inside the substring $ab$.}
We also prove that $b=O(c)$ for internal collage systems, which improves upon that result because they also prove that $\gamma=O(b)$ \cite{KP18}.

\subsection{Suffix Arrays and Runs in the Burrows-Wheeler Transform $(r)$} \label{sec:bwt}

The {\em suffix array} \cite{MM93} of $T[1\dd n]$ is an array $\SA[1\dd n]$ storing
a permutation of $[1\dd n]$ so that, for all $1 \le i < n$, the suffix
$T[\SA[i]\dd ]$ is lexicographically smaller than the suffix
$T[\SA[i+1]\dd ]$. Thus $\SA[i]$ is the starting position in $T$ of the
$i$th smallest suffix of $T$ in lexicographic order. The suffix array can be
built in $O(n)$ time \cite{KSPP05,KA05,KSB06}.

The inverse permutation of $\SA$, $\ISA[1\dd n]$, is called the {\em inverse 
suffix array}, so that $\ISA[j]$ is the lexicographical position of the 
suffix $T[j\dd n]$ among the suffixes of $T$. It can be built in linear time
by inverting the permutation $\SA$.

The {\em longest common prefix array}, $\LCP[1\dd n]$, stores at $\LCP[i]$ the length
of the longest common prefix between $T[\SA[i]\dd ]$ and $T[\SA[i-1]\dd ]$, with
$\LCP[1]=0$. It can be built in linear time from $T$ and $\ISA$ \cite{KLAAP01}. 

The {\em Burrows-Wheeler Transform} of $T$, $\BWT[1\dd n]$ \cite{BW94}, is a
string defined as $\BWT[i] = T[\SA[i]-1]$ if $\SA[i]>1$, and $\BWT[i]=T[n]=\$$ 
if $\SA[i]=1$. That is, $\BWT$ has the same symbols of $T$ in a different order,
and is a reversible transform.

The array $\BWT$ can be easily obtained from $T$ and $\SA$, and thus also be
built in linear time.
To obtain $T$ from $\BWT$ in linear time \cite{BW94}, one considers two arrays,
$L[1\dd n] = \BWT$ and $F[1\dd n]$, which contains all the symbols of $L$ (or $T$) 
in ascending order. Alternatively, $F[i]=T[\SA[i]]$, so $F[i]$ follows $L[i]$ 
in $T$. We need a function that maps any $L[i]$ to the position $j$ of that 
same symbol in $F$. The function is $$\LF(i) ~=~ C[c] + \rank[i],$$ 
where
$c=L[i]$, $C[c]$ is the number of occurrences of symbols less than $c$ in $L$,
and $\rank[i]$ is the number of occurrences of symbol $L[i]$ in $L[1\dd i]$.
Once $C$ and $\rank$ are computed, we reconstruct $T[n]=\$$ and $T[n-k] 
\leftarrow L[\LF^{k-1}(1)]$ for $k=1,\ldots, n-1$.
Note that, if $L[i-1]=L[i]$, then $\LF(i-1)=LF(i)-1$; this
result will be relevant later.

The number $r$ of equal-symbol runs in the BWT of $T$ can be bounded in terms 
of the empirical entropy, $r \le nH_k+\sigma^k$ \cite{MN05}. However, the 
measure
is also interesting on highly repetitive collections (where, in particular,
$z$ and $z_{no}$ are small). For example, it holds $z = \Omega(r\log n)$ 
on the Fibonacci words \cite{Pre16}. However, this result assumes that $T$ is 
not \$-terminated, but that lexicographical 
comparisons take $T$ as a circular string. We will obtain similar results on
our \$-terminated model, which is compatible with the use of $r$ in compressed text indexes. 
On the de Bruijn sequences on binary alphabets, instead, it holds 
$r = \Omega(z_{no}\log n)$ \cite{BCGPR15,Pre16}: we have $r=\Theta(n)$, 
whereas $z_{no}$ is always $O(n/\log n)$ on binary strings.

\begin{example}
The BWT of $T=alabaralalabarda\$$ is 
$L = adll\$lrbbaaraaaaa$, which has $r=10$ runs.
\end{example}

\subsection{Locally consistent parsing} \label{sec:jez}

A string can be parsed in a {\em locally consistent} way,
which means that equal substrings are largely parsed in the same form. 
We use a variant of locally consistent parsing due to Je\.z
\cite{Jez15,I17}.

\begin{definition}
A {\em repetitive area} in a string is a maximal run of the same symbol, of
length 2 or more.
\end{definition}

\begin{definition}
Two intervals contained in $[1\dd n]$ {\em overlap} if they are not disjoint nor 
one contained in the other.
\end{definition}

\begin{definition} \label{def:jez}
A parsing of a string into {\em blocks} is obtained by, first, creating new symbols that represent the
repetitive areas. On the resulting sequence, the alphabet (which contains 
original symbols and created ones) is partitioned into two
subsets, left- and right-symbols. Then, every left-symbol followed by a
right-symbol are paired in a block. The other isolated symbols form a block on their own.
\end{definition}

Je\.z \cite{Jez15} shows that there is a way to choose left- and right-symbols so that the partition into blocks enjoys useful properties, including a form of local consistency.

\begin{lemma}[\cite{Jez15}] \label{lem:lcp}
We can partition a string $S[1\dd\ell]$ into at most $(3/4)\ell$ blocks so that,
for every pair of identical substrings $S[i\dd j] = S[i'\dd j']$, if neither 
$S[i+1\dd j-1]$ nor $S[i'+1\dd j'-1]$ overlap a repetitive area, then the sequence
of blocks covering $S[i+1\dd j-1]$ and $S[i'+1\dd j'-1]$ are identical.
\end{lemma}
\begin{proof}
It is clear that, if $S[i+1\dd j-1]$ and
$S[i'+1\dd j'-1]$ do not overlap repetitive areas, then the parsing of $S[i\dd j]$
and $S[i'\dd j']$ may differ only in their first position (if it is part of a 
repetitive area ending there, or if it is a right-symbol that becomes paired
with the preceding one) and in their last position (if it is part of a 
repetitive area starting there, or if it is a left-symbol that becomes paired
with the following one). Je\.z \cite{Jez15} shows how to choose the pairs so that
$S$ contains at most $(3/4)\ell$ blocks.
\end{proof}

\begin{example}
A locally-consistent parsing of $T=alabaralalabarda\$$ can be 
obtained by considering $a$ to be a left-symbol and all the othes to be right-symbols. The
resulting parsing into blocks is then $T=al|ab|ar|al|al|ab|ar|d|a\$$, where for example in the two
occurrences of $alabar$, the sequence of blocks covering $laba$ are identical, $al|ab|ar$.
\end{example}


\section{Upper Bounds} \label{sec:upper}

In this section we obtain our main upper bound, $z = O(b\log(n/b))$, along with
other byproducts. To this end, we first prove that $g_{rl} = O(b\log(n/b))$, 
and then that $z \le g_{rl}+1$. To prove the first bound, we build an RLSLP
on top of a bidirectional scheme. The grammar is built in several rounds of
locally consistent parsing on the text. In each round, the blocks of the
locally consistent parsing are converted into nonterminals and fed to the next
round. The key is to prove that distinct nonterminals are produced only near the
boundaries of the phrases of the bidirectional scheme.
The second bound is an easy extension of the known result $z_{no} \le g+1$.

\begin{theorem} \label{thm:rlcfg}
There always exists an RLSLP of size $g_{rl} = O(b\log(n/b))$ that
generates $T$.
\end{theorem}
\begin{proof}
Consider the locally consistent parsing of Def.~\ref{def:jez} cutting $W = T$
into blocks. We will count the number of {\em different} blocks we form, as
this corresponds to the number of nonterminals produced in the first round. 

Recall from Section~\ref{sec:bidir} that our bidirectional scheme represents
$T$ as a sequence of {\em phrases}, by means of a function $f$. 
To count the number of different blocks produced, we will pessimistically 
assume that the first two and the last two blocks intersecting each phrase are 
all different. The number of such {\em bordering} blocks is then at most $4b$.
On the other hand, we will show that non-bordering blocks do not need to be 
considered, because they will be counted somewhere else, when they appear 
near the extreme of a phrase. 

\begin{example} 
Let us show how this works on the 
bidirectional scheme example of Section~\ref{sec:bidir},
$al{\bf a|b}|{\bf a|r}|{\bf a|l}|alabar|{\bf d|a|\$}$.
We have selected (in bold) one different block from those created in the example of Section~\ref{sec:jez}. Note that we do not need to select any block that is completely inside a phrase.
We now prove that the general case is only slightly worse.
\end{example} \medskip

We consider both types of non-bordering blocks resulting from Def.~\ref{def:jez}. Figure~\ref{fig:W} illustrates both cases.

\begin{figure}[t]
\includegraphics[width=\textwidth]{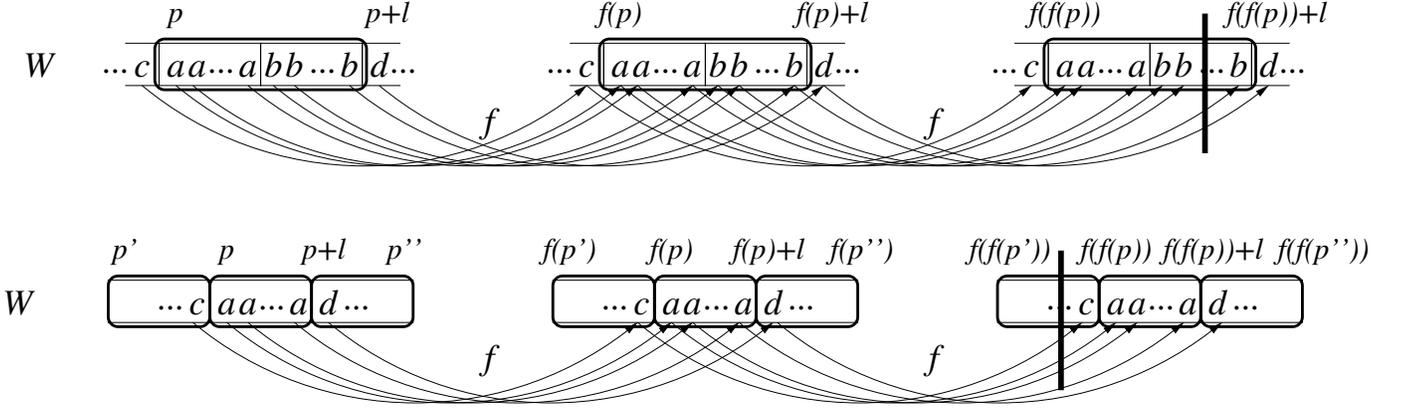}
\caption{The two cases of Theorem~\ref{thm:rlcfg}. On top, case 1, where a block $X= W[p\dd p+\ell-1] =  a^{\ell_a} b^{\ell_b}$ is formed because they are left- and right-symbols surrounded by $c \not= a$ and $d \not= b$. Since all the symbols  are strictly inside a phrase because $X$ is non-bordering, function $f$ maps them together elsewhere in $W$ preserving their contents, so the same block is formed at $W[f(p)\dd f(p)+\ell-1] = X$. This is repeated until a phrase boundary (thick vertical line) appears near $X$ (so that occurrence of $X$ is bordering). On the bottom, case 2, where $X = W[p\dd p+\ell-1] = a^\ell$ is not paired and thus forms a single block surrounded by $c,d \not= a$. Again, the same contents are found, and the same blocks are formed, at $W[f(p)\dd f(p)+\ell-1] = X$ because the blocks $Y=W[p'\dd p-1]$ and $Z=W[p+\ell\dd p'']$ are strictly inside a phrase. Again, this is repeated until hitting a phrase boundary nearby.}
\label{fig:W}
\end{figure}

\begin{enumerate}
\item The block is a pair of left- and right-alphabet symbols.%
\footnote{For this case, we could have defined {\em bordering} in a stricter 
way, as the first or last block of a phrase.} 
As these symbols
can be an original symbol or a repetitive area, let us write the pair generically 
as $X = a^{\ell_a} b^{\ell_b}$, for some $\ell_a,\ell_b \ge 1$, and let 
$\ell=\ell_a +\ell_b$ be the length of the block $X$. If $W[p\dd p+\ell-1]=X$ 
is not bordering, then it is strictly contained in a phrase. Thus, by the 
definition of a phrase, it 
holds that $[f(p-1)\dd f(p+\ell)]=[f(p)-1\dd f(p)+\ell]$, and that
$W[f(p)-1\dd f(p)+\ell] = W[p-1\dd p+\ell]$. That is, the block appears again at 
$[f(p)\dd f(p)+\ell-1]$, surrounded by the same symbols. Since Def.~\ref{def:jez}
first compacts repetitive areas, it must be $W[f(p)-1] = W[p-1] \not= a$ and 
$W[f(p)+\ell]=W[p+\ell] \not= b$. Further, since Def.~\ref{def:jez} pairs left- with
right-symbols, $a^{\ell_a}$ must be a left-symbol and
$b^{\ell_b}$ must be a right-symbol. The locally consistent parsing must then also form 
a block $W[f(p)\dd f(p)+\ell-1] = X$. If this block is bordering, then it will 
be counted. Otherwise, by the same 
argument, $W[f(p)-1\dd f(p)+\ell]$ will be equal to $W[f^2(p)-1\dd f^2(p)+\ell]$ 
and a block will be formed with $W[f^2(p)\dd f^2(p)+\ell-1]$. 
Since $f$ has no cycles, there is a $k>0$ for which $f^k(p)=-1$. Thus for some 
$l\le k$ it must be that $X=W[f^l(p)\dd f^l(p)+\ell-1]$ is bordering.
At the smallest such $l$, the block $W[f^l(p)\dd f^l(p)+\ell-1]$ will 
be counted. Therefore, $X = W[p\dd p+\ell-1]$ is already counted somewhere else.

\item The block is a single (original or maximal-run) symbol 
$W[p\dd p+\ell-1] = a^\ell = X$, for some $\ell \ge 1$. It also holds that 
$[f(p-1)\dd f(p+\ell)]=[f(p)-1\dd f(p)+\ell]$ and 
$W[f(p)-1\dd f(p)+\ell]=W[p-1\dd p+\ell]$, because $a^\ell$ is strictly inside a
phrase. Since $W[f(p)-1]=W[p-1] \not= a$ and $W[f(p)+\ell]=W[p+\ell]\not=a$,
the parsing forms the same maximal run $X = a^\ell = W[f(p)\dd f(p)+\ell-1]$.
Moreover, since $W[p\dd p+\ell-1]$ is not bordering, the previous and next
blocks produced by the parsing, $Y=W[p'\dd p-1]$ and $Z=[p+\ell\dd p'']$, are also 
strictly inside the same phrase, and therefore they also appear preceding and
following $W[f(p)\dd f(p)+\ell-1]$, at $Y=W[f(p')\dd f(p)-1]$ and 
$Z=[f(p)+\ell\dd f(p'')]$. Since $a^\ell$ was not paired with $Y$ nor $Z$ at
$W[p\dd p+\ell-1]$, the parsing will also not pair them at $W[f(p)\dd f(p)+\ell-1]$.
Therefore, the parsing will leave $a^\ell$ as a block also in 
$[f(p)\dd f(p)+\ell-1]$. If $W[f(p)\dd f(p+\ell-1)]$ is bordering, then it will be 
counted, otherwise we can repeat the argument with $W[f^2(p)-1\dd f^2(p)+\ell]$
and so on, as before.
\end{enumerate}

Therefore, we produce at most $4b$ distinct blocks, and the RLSLP has at most
$12b$ nonterminals (for $X = a^{\ell_a} b^{\ell_b}$ we may need 3 nonterminals,
$A \rightarrow a^{\ell_a}$, $B \rightarrow b^{\ell_b}$, and $C \rightarrow AB$).

For the second round, we create a reduced sequence $W'$ from $W$ by replacing
all the blocks of length $2$ or more by their corresponding nonterminals. The 
new sequence is guaranteed to have length at most $(3/4)n$ by 
Lemma~\ref{lem:lcp}. 

We then define a new bidirectional scheme (recall Section~\ref{sec:bidir}) on $W'$, 
as follows:
\begin{enumerate}
\item For each bordering block in $W$, its nonterminal symbol position in $W'$
is made explicit in the bidirectional scheme of $W'$. Note that this includes
the blocks covering the explicit symbols in the bidirectional scheme of $W$.
\item For the phrases $B_i=W[t_i\dd t_i+\ell_i-1]$ of $W$ containing non-bordering
blocks (note $B_i$ cannot be an explicit symbol), let $B'_i$ be obtained by 
trimming from $B_i$ the bordering blocks near the boundaries of $B_i$.
Then $B'_i$ appears inside $W[s_i\dd s_i+\ell_i-1]$ (with $s_i=f(t_i)$), where 
the same sequence of blocks is formed by our arguments above. We then form a
phrase in $W'$ with the sequence of nonterminals associated with the blocks of 
$B_i'$ (all of which are non-bordering), pointing to the identical sequence of 
nonterminals that appear as blocks inside $W[s_i\dd s_i+\ell_i-1]$.
\end{enumerate}

\begin{example} 
On our preceding example,
$al{\bf a|b}|{\bf a|r}|{\bf a|l}|alabar|{\bf d|a|\$}$, we define the nonterminals
$A \rightarrow ab$, $B \rightarrow ar$, $C \rightarrow al$, and $D \rightarrow a\$$.
We then create $W' = C|\underline{A}|\underline{B}|\underline{C}|CAB|\underline{d}|\underline{D}$,
where we show the derived bidirectional scheme and underline the explicit symbols. Recall that, to make this
small example interesting, we have been stricter about which blocks are bordering.
\end{example}\medskip

To bound the total number of nonterminals generated, let us call $W_k$ the 
sequence $W$ after $k$ iterations (so $T=W_0$) and $N_k$ the number of 
distinct blocks created when converting $W_k$ into $W_{k+1}$. 

In the first iteration, since there may be up to $4$ bordering blocks around 
each phrase limit, we may create $N_1 \le 4b$ distinct blocks. Those blocks
become new explicit symbols in the bidirectional scheme of $W' = W_1$. Note that
those explicit symbols are grouped into $b$ {\em regions} of up to $4$ 
consecutive symbols. In each new iteration, $W_k$ is parsed into blocks again. 
We have shown that the non-bordering blocks formed are not distinct, so we 
can focus on the number of new blocks produced to parse each of the $b$ regions and
near their surrounding phrase boundaries. The parsing produces at most $4$ new 
distinct blocks extending each region (i.e., at the phrase boundaries surrounding the
region). 
However, the parsing of the regions themselves may also produce new distinct 
blocks. Our aim is to show that the number of those blocks is also
bounded because they decrease the length of the regions, which only grow by
$4b$ (explicit symbols) per iteration. Intuitively, each new nonterminal created
inside a region decreases its length, and thus both numbers cancel out. We now make 
the argument more precise.

\begin{figure}[t]
\centering
\includegraphics[width=0.6\textwidth]{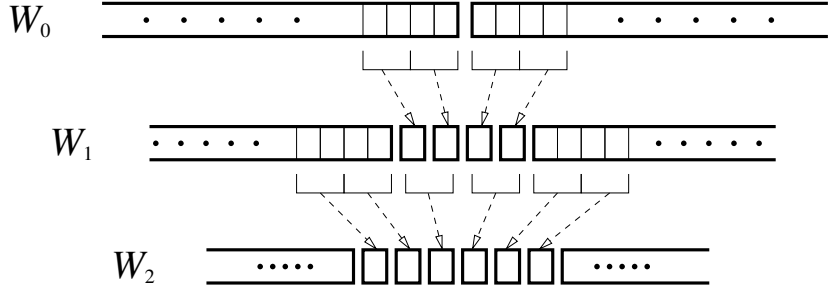}
\caption{Illustration of Theorem~\ref{thm:rlcfg}. On top we see the limit
between two long phrases of $W_0$. In this example, the blocking always pairs 
two symbols. We show below $W_0$ the $4$ bordering blocks formed with the
symbols nearby the boundary. Below, in $W_1$, those blocks are converted into $4$ 
explicit symmbols. This region of $4$ symbols is then parsed into $2$ blocks.
The parsing also creates $4$ new bordering blocks from the boundaries of the long
phrases. In $W_2$, below, we have now a region of $6$ explicit symbols. They 
would have been $8$, but we created $2$ distinct blocks that reduced 
their number to $6$.}
\label{fig:parsing}
\end{figure}

Let $n_k$ be the number of new distinct blocks produced when parsing the 
regions themselves. Therefore it holds that the number $N_k$ of distinct blocks 
produced in the $k$th iteration is at most $4b + n_k$, and the total 
number of distinct blocks created up to building $W_k$ is
$$\sum_{i = 0}^{k - 1} N_i ~\le~ 4 b k + \sum_{i = 0}^{k-1} n_i.$$

On the other hand, for each of the $n_k$ blocks created when parsing a region, 
the length of the region decreases at least by $1$ in $W_{k+1}$, that is, there is one
explicit symbol less in $W_{k+1}$. Let us call 
$C_k$ the number of explicit symbols in $W_k$. Since only the $4$ new bordering 
blocks surrounding each region are converted into explicit symbols when creating $W_k$, it holds that 
$C_k \le 4bk$ for all $k>0$. Moreover, it holds that $C_{k+1} \le C_k + 4b - n_k$, 
and thus $$0 ~\le~ C_k ~\le~ 4bk - \sum_{i=0}^{k-1} n_i.$$ 
It follows that  
$\sum_{i=0}^{k-1} n_i \le 4bk$, and thus $$\sum_{i=0}^{k-1} N_i ~\le~ 8bk.$$
Since each nonterminal may need $3$ rules to represent a block, a bound on 
the number of nonterminals created is $24bk$.
The idea is illustrated in Figure~\ref{fig:parsing}.

After $k$ rounds, the sequence is of length at most $(3/4)^k n$ and we have
generated at most $24bk$ nonterminals. Therefore, if we choose to perform 
$k = \log_{4/3}(n/b)$ rounds, the sequence will be of length at most $b$ and 
the RLSLP size will be $O(b\log(n/b))$. To complete the process, we add 
$O(b)$ nonterminals to reduce the sequence to a single initial symbol.
\end{proof}

With Theorem~\ref{thm:rlcfg}, we can also bound the size $z$ of the Lempel-Ziv
parse \cite{LZ76} that allows overlaps. The size without allowing overlaps is 
known to be bounded by the size of the smallest SLP, $z_{no} \le g+1$
\cite{Ryt03,CLLPPSS05}. We can easily see that $z \le g_{rl}+1$ also holds by 
extending an existing proof \cite[Lem.~9]{CLLPPSS05} to handle the run-length 
rules. 
We call any parsing of $T$ where every new phrase is a symbol or it occurs previously in $T$ a
\emph{left-to-right parse}.

\begin{theorem} \label{thm:grlz}
The Lempel-Ziv parse 
(allowing overlaps) of $T$ always produces $z \le g_{rl}+1$ phrases.
\end{theorem}
\begin{proof}
Consider the grammar tree of $T$ (Section~\ref{sec:gram}), where only the leftmost occurrence of each 
nonterminal $X$ is an internal node. Our left-to-right  parse of $T$ is a 
sequence $Z[1\dd z]$ obtained by traversing the leaves of the grammar tree 
left to right. For a terminal leaf, we append the explicit symbol to $Z$. For a 
leaf representing nonterminal $X$, we append to $Z$ a reference to the area 
$T[x\dd y]$ expanded by the leftmost occurrence of $X$. 

To extend grammar trees to RLSLPs, we handle rules $X \rightarrow Y^t$ as follows. 
First, we expand them to
$X \rightarrow Y\cdot Y^{t-1}$, that is, the node for $X$ has two children
for $Y$, the second annotated with $t-1$.
Since the right child of $X$ is not the first occurrence of $Y$, it must be a
leaf. The left child of $X$ may or may not be a leaf, depending on whether
$Y$ occurred before or not. Since run-length rules become internal nodes with two children, it still holds that the grammar tree has at most $g_{rl}+1$ leaves.

Now, when our leaf traversal reaches the right
child $Y$ of a node $X$ indicating $t-1$ repetitions, we append to $Z$ a
reference to $T[x\dd y+(t-2)(y-x+1)]$, where $T[x\dd y]$ is the area expanded by
the first child of $X$. Note that source and target overlap if $t > 2$.
Thus a left-to-right parse of size $g_{rl}+1$ exists, and the result follows because
Lempel-Ziv is the optimal left-to-right parse \cite[Thm.~1]{LZ76}.
\end{proof}

By combining Theorems~\ref{thm:rlcfg} and \ref{thm:grlz}, we obtain a result 
on the long-standing open problem of finding the
approximation ratio of Lempel-Ziv compared to the smallest bidirectional
scheme.

\begin{theorem} \label{thm:main}
The Lempel-Ziv 
parsing of $T$ allowing overlaps always has $z = O(b \log(n/b))$ phrases.
\end{theorem}

We can also derive upper bounds for $g$ and $z_{no}$.
It is sufficient to combine Theorem~\ref{thm:main} with the facts that 
$g = O(z\log(n/z))$ \cite[Lem.~8]{Gaw11} and 
$z_{no} \le g+1$ \cite{Ryt03,CLLPPSS05}.

\begin{lemma} 
It always holds that $g,z_{no} = O(b\log^2(n/b))$.
\end{lemma}

%
%


\section{Lower Bounds} \label{sec:lower}

In this section we prove that the upper bound of Theorem~\ref{thm:main} is tight
as a function of $n$, by exhibiting a family of strings for which
$z = \Omega(b \log n)$. This confirms that the gap between bidirectionality 
and unidirectionality is significantly larger than what was previously known. 
The idea is to define phrases in $T$ according to the $r$ runs in the BWT, 
and to show that these phrases induce a valid bidirectional scheme of 
size $2r$. This proves that $r = \Omega(b)$. Then we use a well-known family
of strings where $z = \Omega (r \log n)$.

\begin{definition} \label{def:bwtparse}
Let $p_1, p_2, \ldots, p_r$ be the positions that start runs in the BWT, and
let $t_1 < t_2 < \ldots < t_r$ be the corresponding positions in $T$,
$\{ \SA[p_i] \mid 1 \le i \le r \}$, in increasing order. Note that $t_1=1$ because
$\BWT[\ISA[1]]=\$$ is a size-$1$ run, and let $t_{r+1}=n+1$, so that $T$ is 
partitioned into {\em phrases} $T[t_i\dd t_{i+1}-1]$. 
Let also $\phi(p)=\SA[\ISA[p]-1]$ if
$\ISA[p]>1$ and $\phi(p)=\SA[n]$ otherwise. Then we define the {\em 
bidirectional scheme of the BWT:}
\begin{enumerate}
\item For each $1 \le i \le r$, $T[t_i\dd t_{i+1}-2]$ is copied from
$T[\phi(t_i)\dd \phi(t_{i+1}-2)]$.
\item For each $1 \le i \le r$, $T[t_{i+1}-1]$ is an explicit symbol.
\end{enumerate}
\end{definition}

\begin{example}
The BWT runs of the example of Section~\ref{sec:bwt} induces
the bidirectional scheme
$\underline{a}|\underline{l}|\underline{a}|\underline{b}|a|\underline{r}|a|\underline{l}|alaba|\underline{r}|\underline{d}|\underline{a}|\underline{\$}$.
\end{example}\medskip

We build on the following lemma, illustrated in Figure~\ref{fig:phi}. We make use of the function $\LF$ defined in Section~\ref{sec:bwt}. Note that $\LF(x)=\ISA[\SA[x]-1]$ if $\SA[x]>1$ and $\LF(x)=\ISA[n] = 1$ if $\SA[x]=1$. That is, $\LF$ moves in $\SA$ to the suffix preceding the current one in $T$. The analogous function moving in $T$ to the suffix preceding the current one in $\SA$ is $\phi$.

\begin{figure}[t]
\centering
\includegraphics[width=0.5\textwidth]{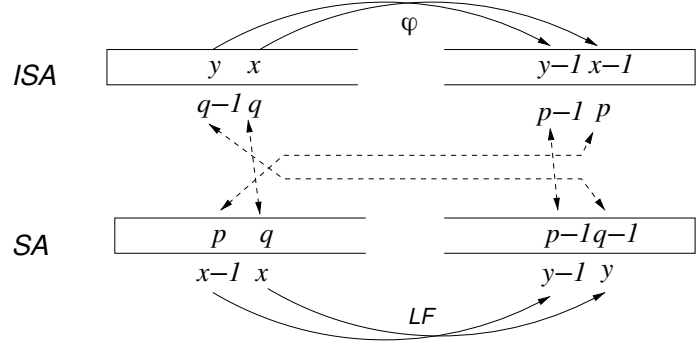}
\caption{Illustration of Lemma~\ref{lem:phi}.}
\label{fig:phi}
\end{figure}

\begin{lemma} \label{lem:phi}
Let $[q-1\dd q]$ be within a phrase of $T$. Then it holds that $\phi(q-1)=\phi(q)-1$
and $T[q-1]=T[\phi(q)-1]$.
\end{lemma}
\begin{proof}
Consider the pair of positions $T[q-1\dd q]$ within a phrase. Let them be pointed 
from $\SA[x]=q$ and $\SA[y]=q-1$, therefore $\ISA[q] = x$, $\ISA[q-1]=y$,
and $\LF(x)=y$. Now, since $q$ is not a position at the beginning of a phrase, $x$ is not the first 
position in a BWT run. Therefore, $\BWT[x-1]=\BWT[x]$. Recalling the formula of Section~\ref{sec:bwt} to compute $\LF(x)=C[c]+\rank[x]$, where $c=\BWT[x]$, 
it follows that $\LF(x-1)=\LF(x)-1=y-1$. Now let $\SA[x-1]=p$, that is, 
$p=\phi(q)$. Then, $$\phi(q-1)~=~\SA[\ISA[q-1]-1]~=~\SA[y-1]~=~\SA[\LF(x-1)]~=~\SA[x-1]-1
~=~p-1~=~\phi(q)-1.$$
It also follows that $$T[q-1]~=~\BWT[x]~=~ \BWT[x-1]~=~T[p-1]~=~T[\phi(q)-1].$$
\end{proof}

\begin{example}
The suffix array of $T=alabaralalabarda\$$ is $\SA=\langle 17,16,3,11,1,9,7,5,13,4,12,15,2,10,8,6,14\rangle$ and the $\phi$ function is $\phi = \langle 11,15,16,13,7,8,9,10,1,2,3,4,5,6,12,17,14\rangle$. For example, $\phi(1)=11$ because the suffix lexicographically preceding $T[1\dd]$ is $T[11\dd]$. Now, let $q=10$, which is inside the same phrase of $q-1=9$ in the parse $\underline{a}|\underline{l}|\underline{a}|\underline{b}|a|\underline{r}|a|\underline{l}|alaba|\underline{r}|\underline{d}|\underline{a}|\underline{\$}$ induced by the run heads of the BWT of $T$, $BWT = adll\$lrbbaaraaaaa$. Position $T[q=10]$ is pointed from $\SA[x=14]$, whereas $T[q-1=9]$ is pointed from $\SA[y=6]$. Thus $\LF(x=14) = C[\BWT[14]=a] + \rank[14] = 1+5=6=y$. Since $q=10$ does not start a phrase in $T$, $\BWT[x=14]$ does not start a run, thus $\BWT[x-1=13]=a$. It then holds that $\LF(x-1=13)=C[\BWT[13]=a]+\rank[13] = 1+4 = 5 = y-1 =\LF(x=14)-1$. Further, if we call $p=\SA[x-1=13]=2$, it holds that $p=2=\phi(q=10)$. One can then verify that $\phi(q-1=9) = \SA[y-1=5] = 1 = \SA[x-1=13]-1=\phi(q=10)-1$, and that $T[q-1=9] = a = \BWT[x=14] = \BWT[x-1=13] = T[p-1=1] = T[\phi(q=10)-1]$.
\end{example}

\begin{theorem} \label{thm:r}
The bidirectional scheme of the BWT is a valid bidirectional scheme, thus 
it always holds $b \le 2r$.
\end{theorem}
\begin{proof}
By Lemma~\ref{lem:phi}, it holds that $\phi(q-1)=\phi(q)-1$ if $[q-1\dd q]$ is
within a phrase, and that $T[q-1]=T[\phi(q)-1]$.
Therefore, we have that $\phi(t_i+k)=\phi(t_i)+k$ for $0 \le k < \ell_i =
t_{i+1}-t_i-1$, and then $T[\phi(t_i)\dd \phi(t_{i+1}-2)]$ is indeed a 
contiguous range of length $\ell_i$. We also have that 
$T[\phi(t_i)\dd \phi(t_{i+1}-2)] = T[t_i\dd t_{i+1}-2]$, and therefore the copy is
correct. 

It is also easy to see that we can recover the whole $T$ from those $2r$
phrases. We can, for example, follow the cycle $\phi^k(n)$, $k=n-1,\ldots,1$,
and copy $T[\phi^k(n)]$ from $T[\phi^{k+1}(n)]$ unless the former is explicitly
stored (note that $T[\phi^n(n)]=T[\phi^0(n)]=T[n]$ is stored explicitly). By
Lemma~\ref{lem:phi}, it is correct to copy from $T[\phi(p)]$ to $T[p]$ whenever 
$p$ (which is $q-1$ in Lemma~\ref{lem:phi}) is not at the end of a phrase; this is why
we store the explicit symbols at the end of the phrases. 

Since the bidirectional scheme of the BWT is of size $2r$, it follows by
definition that $2r \ge b$.
\end{proof}

\begin{example}
We can recover $T$ from our bidirectional scheme
$\underline{a}|\underline{l}|\underline{a}|\underline{b}|a|\underline{r}|a|\underline{l}|alaba|\underline{r}|\underline{d}|\underline{a}|\underline{\$}$ by following positions $\phi^{n-1}(n),\ldots,$ $\phi(n)$
and copying the last explicit symbol seen onto each new position. The sequence, where we indicate in parentheses the explicit symbols visited, is 
$16(a),3(a),11,1(a),9,7,5,13,4(b),12,15(d),2(l),10,8(l),6(r),14(r)$. For example, the explicit 
$a$ collected at $T[1]$ is copied onto $T[9]$, $T[7]$, $T[5]$, and $T[13]$.
\end{example}\medskip

We can now prove the promised separation betweeen $z$ and $b$.
Before, we prove a further property of the cyclic rotations
of the Fibonacci words we make use of.

\begin{lemma}
  \label{lem:smallest}
  In every {\em even} Fibonacci word $F_k$,
  the lexicographically smallest cyclic
  rotation is the one that starts at the last character.
\end{lemma}
 
\begin{proof}
  Mantaci et al.~\cite{MRS07} give a characterization of
  the cyclic rotations of the $k$th Fibonacci word $F_k$ by defining
  two functions: $\varrho: [0 \dd f_k - 1] \rightarrow [0 \dd f_k-1]$,
  defined as
  \[\varrho(x) = x + f_{k-2} ~~(\bmod~ f_k),\]
  and $\varphi: [0 \dd f_k -1] \rightarrow \{a, b\}$,
  defined as
  \[
    \varphi(x) =
    \begin{cases} 
      a, & \textrm{if~} x < f_{k-1} \\
      b, & \textrm{if~} x \ge f_{k-1},
    \end{cases}
  \]
where they index the strings from position $0$.
  They prove that the cyclic rotations of $F_k$ are
  the words $R_x = r_0r_1\cdots r_{f_k - 1}$, where
  $r_i = \varphi(\varrho^i(x))$, for $0 \le x \le f_k -1$.
  The lexicographic ordering of the cyclic rotations of $F_k$ is
  $R_0 < R_1 < \cdots < R_{f_k-1}$~\cite[proof of Thm.~9]{MRS07}.
  If $k$ is even, then $F_k = R_{f_{k-2}}$~\cite[Thm.~6]{MRS07}.
 Then, since $F_k[i] = R_{f_{k-2}}[i] = \phi(\varrho^i(f_{k-2})) =
 \phi(\varrho^{i+1}(0)) = R_0[i+1]$, and $R_0$ is the lexicographically
 smallest cyclic rotation, the lexicographically smallest cyclic rotation of
 $F_k$ starts at its last position, $f_k$.
  Formally, $F_k[f_k]F_k[1\dd f_k-1]$ is the lexicographically
  smallest cyclic rotation of $F_k$, for all even $k$.
\end{proof}

\begin{theorem} \label{thm:lowz}
There is an infinite family of strings over an alphabet of size 2 for which $r=O(1)$ and $z=\Theta(\log n)$, and thus
$z = \Omega(r \log n)$ and
 $z = \Omega(b \log n)$.
\end{theorem}
\begin{proof}
As observed by Prezza \cite[Thm.~25]{Pre16}, for all Fibonacci words we have 
$r=O(1)$ \cite[Thm.~9]{MRS07}. Combining it with the fact that, in all 
Fibonacci words, it holds $z=\Theta(\log n)$ \cite{Fic15}, yields 
$z = \Omega(r \log n)$.

Note, however, that the result $r=O(1)$ is proved under a BWT definition that is 
different from ours \cite{MRS07}. Namely, the Fibonacci words are not
terminated with $\$$, but instead the suffixes are compared cyclically, as if 
$F_k$ were a circular word.

By Lemma~\ref{lem:smallest}, however, in each 
{\em even} Fibonacci word $F_k$, the lexicographically smallest cyclic suffix
is the one that starts at the last character.
From this observation we have that, in every even Fibonacci word $F_k$,
the relative order of the cyclic suffixes is the same as the relative order
of the suffixes terminated in $\$$. Formally,
$F_k[i\dd f_k]F_k[1\dd i-1] < F_k[j\dd f_k]F_k[1\dd j-1]$
if and only if
$F_k[i\dd f_k]\$ < F_k[j\dd f_k]\$$, for all $i \neq j$, and $k$ even.
Thus, in the even Fibonacci words, we have $r=O(1)$, and thus $z=\Omega(r\log n)$.
The result $z=\Omega(b\log n)$ then  follows from the fact that $b=O(r)$, by
Theorem~\ref{thm:r}.
\end{proof}

Finally, by relating $g$ with the empirical entropy of $T$, we
can also prove a separation between $r$ and $g$.

\begin{lemma}\label{lem:entropy}
It always holds that 
        $g\log_2 n \leq nH_k + o(n\log\sigma)$ for any $k = o(\log_\sigma n)$,
thus $g = O(n/\log_\sigma n)$.
\end{lemma}
\begin{proof}
	Let $z_{78}$ be the size of the Lempel-Ziv 1978 (LZ78) parsing 
\cite{ZL78} of $T$. Then, it holds that $z_{78}\log_2 n \leq nH_k + 
o(n\log\sigma)$ for $k=o(\log_\sigma n)$ \cite[Thm.~A.4]{KM00} (noting 
that their $c$ is $O(n/\log_\sigma n)$ and assuming $k=o(\log_\sigma n)$).
Since this parsing can be converted into an SLP of size $z_{78}$,
it holds that $g \le z_{78}$ and the result follows.
The final claim is a consequence of the fact that $H_k \le \log\sigma$ for all $k$.
\end{proof}

\begin{theorem}\label{thm:lowrg}
There is an infinite family of strings over an alphabet of size 2 for which
$r = \Omega(g \log n)$.
\end{theorem}
\begin{proof}
By Lemma~\ref{lem:entropy},
the smallest SLP on a binary alphabet is always of size $g = O(n/\log n)$.
On a de Bruijn sequence of order $k$ on a binary alphabet we have 
$r=\Theta(n)$ \cite{BCGPR15}. The result follows.
\end{proof}


\section{Greedy and Ordered Parses}

In this section we extend the Lempel-Ziv parse, where sources must start before
targets in the text, to the more general concepts of ordered parsings, and 
prove some general results on them.


\begin{definition} \label{def:ordered}
An \emph{ordered parse} of $T[1\dd n]$ is a partition of $T$
into $u$ phrases $B_1,\ldots,B_u$, such that each phrase 
$B_i = T[t_i\dd t_i+\ell_i-1]$ either is an explicit symbol or it is
copied from a source $T[s_i\dd s_i+\ell_i-1]$, such that $s_i \not= t_i$ and
$T[s_i+j\dd ]\prec T[t_i+j\dd ]$ for all $0 \le j < \ell_i$, for some suitable 
total order $\preceq$ on the suffixes of $T$.\footnote{The order is called $\preceq$ because it must be reflexive, yet we use $x \prec y$ to indicate $x \preceq y$ and $x \not= y$, that is, $x$ is strictly smaller than $y$ under order $\preceq$.}
\end{definition}

By the way we define them, ordered parses are bound to be valid bidirectional schemes, and
bidirectional schemes are ordered parses under some suitable order.

\begin{lemma} \label{lem:ordered-to-b}
Every ordered parse is a bidirectional scheme.
\end{lemma}
\begin{proof}
Let $f$ be the function associated with the ordered parse, that is, $f(t_i+j) = s_i+j$ for all $0 \le j < \ell_i$ if phrase $B_i = T[t_i\dd t_i+\ell_i-1]$ is copied from $T[s_i\dd s_i+\ell_i-1]$.
There cannot be a cycle in $f$ because, by definition, $T[f(p)\dd] \prec 
T[p\dd]$ for every position $p$ inside every such phrase $B_i$.
\end{proof}


\begin{lemma} \label{lem:b-to-ordered}
Every bidirectional scheme is an ordered parse under some suitable order 
$\preceq$.
\end{lemma}
\begin{proof}
Let $f$ be the function associated with the bidirectional scheme. Let us assign
to every suffix $T[p\dd]$ the value $h(p) = \min \{k \ge 0, f^k(p)=-1\}$.
Now $\preceq$ can be any total order on $[1\dd n]$ compatible with $h(p)$,
that is, such that if $h(p') < h(p)$ then $p' \prec p$ (e.g., topological sorting
produces a valid order $\preceq$). Since the bidirectional
scheme copies $T[p]$ from $T[f(p)]$ and $h(p) = 1+h(f(p)) > h(f(p))$, it
holds that $T[f(p)\dd] \prec T[p\dd]$. The parsing is then ordered under order
$\preceq$.
\end{proof}

We are interested in parses that, while respecting a given order $\preceq$, 
produce the smallest number of phrases.

\begin{definition} \label{def:optimal-parse}
A parse is {\em ordered-optimal} with respect to a total
order $\preceq$ if no other ordered parse respecting the 
order $\preceq$ has fewer phrases on any text $T[1\dd n]$. We may or may
not allow that sources and targets overlap to define optimality.
\end{definition}

Lempel-Ziv is an ordered parse with respect to the order 
$T[s_i\dd ] \prec T[t_i\dd ]$ defined as $s_i < t_i$. The parses that respect
this order are called left-to-right parses. 
As we have seen, then, Lempel-Ziv is 
ordered-optimal, either with or without overlaps \cite{LZ76,SS82}.
Further, the methods that obtain those optimal parses \cite{RPE81,CIKRW12} 
are {\em greedy}, under the following definition.

\begin{definition} \label{def:greedy}
A method to obtain an ordered parse of $T[1\dd n]$ is 
{\em greedy} if it proceeds left to right on $T$ producing one phrase at each 
step, and such phrase is the longest possible one that starts at that position 
and has a smaller source in $T$ under the order $\preceq$. If the longest possible
phrase is of length $0$ or $1$, the parse may use an explicit symbol.
\end{definition}

Greedy methods are attractive on ordered parses because they produce the 
ordered-optimal parse and can be computed in polynomial time.

\begin{theorem}\label{thm:greedy}
Every greedy parse is ordered-optimal.
\end{theorem}
\begin{proof}
Let $B_1,\ldots,B_u$ be the result of the greedy parsing of $T$ under order 
$\preceq$. Since the first phrase always starts at position 1, if there is 
another ordered parse $B_i'=T[t_i'\dd t_i'+ \ell_i'-1]$ for $1 \le i \le u'$ 
and $u'<u$, then there must be a first phrase where $t_{i+1}'>t_{i+1}$. Since 
it is the first, it must hold that $t_i' \le t_i < t_{i+1} < t'_{i+1}$. Let us 
call $\delta = t_i - t_i' < \ell_i' = t_{i+1}'-t_i'$. Therefore, there is a 
source $T[s_i'\dd s_i'+\ell'_i-1]=T[t_i'\dd t_i'+\ell'_i-1]$ such that 
$T[s_i'+j\dd ] \prec T[t_i'+j\dd ]$ for all $0 \le j < \ell'_i$. Then it also
holds that $T[s_i'+\delta\dd s_i'+\ell'_i-1] = T[t_i\dd t_{i+1}'-1]$ and that 
$T[s_i'+\delta+j\dd ] \prec T[t_i+j\dd ]$ for all $0 \le j < t_{i+1}'-t_i$. 
Therefore, there exists a suffix $T[s_i'+\delta\dd ]$ that shares with 
$T[t_i\dd ]$ a prefix of length $t_{i+1}'-t_i > t_{i+1}-t_i = \ell_i$ and it
can be used under order $\preceq$. This is impossible because our parsing is
greedy and it did not choose that suffix when producing the phrase $T[t_i\dd]$.
\end{proof}

\begin{theorem} \label{thm:ncube}
The greedy parse of any ordered parse can be obtained in $O(n^3)$ evaluations 
of $\prec$.
\end{theorem}
\begin{proof}
We obtain the phrase lengths $\ell_i$ left to right, so that their starting
points are $s_1 = 1$ and $s_{i+1} = s_i + \ell_i$. To find the length $\ell_i$ 
of each new phrase $T[s_i \dd s_i+\ell_i-1]$, we compare the suffix $T[s_i \dd]$
with every other suffix $T[p \dd]$ symbol by symbol, until we find the smallest
$j_p \ge 0$ such that $T[p + j_p] \succ T[s_i + j_p]$ or $p+j_p>n$ or 
$s_i+j_p>n$. We then have $\ell_i = \max_{p \not= s_i} j_p$. If $j_i=0$ we
create an explicit symbol in the parse.
%
\end{proof}

Of course, particular greedy parses, like Lempel-Ziv, can be obtained faster,
in this case in time $O(n)$ \cite{RPE81,CIKRW12}.
%
Interestingly, the fact that ordered-optimal parses are computed easily implies that we cannot
efficiently find the order $\preceq$ that produces the smallest ordered parse.

\begin{lemma}
Finding the order $\preceq$ that produces the smallest ordered parse on $T$
is NP-hard.
\end{lemma}
\begin{proof}
One of those orders $\preceq$ yields the smallest bidirectional scheme,
by Lemma~\ref{lem:b-to-ordered}. Once we have the best order $\preceq$, we
find the parse itself greedily in time $O(n^3)$, by Theorems~\ref{thm:greedy}
and \ref{thm:ncube}.
Thus we obtain the smallest bidirectional scheme, which is NP-hard to find
\cite{Gal82}.
\end{proof}

On the other hand, we now show that, under certain favorable kinds of orders $\preceq$, 
the size of the ordered-optimal parses is upper bounded by the size 
of the smallest grammar. 
In particular, ordered-optimal parses that let sources
and targets overlap are of size $O(b\log(n/b))$.

\begin{definition}
A total order $\preceq$ on text suffixes is {\em extensible} if $T[s\dd] \prec T[t\dd]$ and
$T[s] = T[t]$ implies that $T[s+1\dd] \prec T[t+1\dd]$.
\end{definition}

For example, the order of left-to-right parses, $T[s\dd] \prec T[t\dd]$ iff $s<t$, is extensible.

\begin{theorem} \label{thm:ordered-g}
Any ordered-optimal parse of $T$, for any extensible order $\preceq$, produces
$u \le g+1$ phrases. Thus, $u\log_2 n \le nH_k + o(n\log\sigma)$ for any
$k=o(\log_\sigma n)$, $u = O(n/\log_\sigma n)$, and there are string families
where $r = \Omega(u\log n)$.
\end{theorem}
\begin{proof}
It suffices to show how to build an ordered parse of size at most $g+1$.
Analogously to the proof of Theorem~\ref{thm:grlz}, consider a variant of the grammar
tree of $T$ where the only internal node labeled $X$ and expanding to
$T[x_i\dd z_i]$ is the one with the smallest suffix
$T[x_i\dd ]$ under order $\preceq$. This tree has up to $g+1$ leaves, just like the original grammar tree.
We then define an ordered parse of $T$ by converting every terminal leaf
to an explicit symbol, and every leaf covering $T[x_i'\dd z_i']$, labeled by nonterminal $X$,
to a phrase that points to the area $T[x_i\dd z_i]$ corresponding to the only internal
node labeled $X$. Such a parse is of size $u \le g+1$ and is ordered because the
order is extensible: since $T[x_i'\dd z_i'] = T[x_i\dd z_i]$ and $T[x_i\dd] \prec T[x_i'\dd]$,
it follows that $T[x_i+j\dd] \prec T[x_i'+j\dd]$ for all $0\le j\le z_i-x_i$.

Since this is an ordered parse, the ordered-optimal parse is also of size $u \le g+1$.
The other results are immediate consequences of
Lemma~\ref{lem:entropy} and Theorem~\ref{thm:lowrg}.
\end{proof}

\begin{theorem} \label{thm:ordered-grl}
Any ordered-optimal
parse of $T$ that allows sources and targets overlap, under any extensible order $\preceq$,
produces $u \le g_{rl}+1$ phrases. Thus it holds that $u = O(b\log(n/b))$.
\end{theorem}
\begin{proof}
We extend the proof of Theorem~\ref{thm:ordered-g} so as to consider the rules
$X \rightarrow Y^t$. These can be expanded either to
$X \rightarrow Y\cdot Y^{t-1}$ or to $X \rightarrow Y^{t-1} \cdot Y$.
In both cases, the child $Y$ is handled as usual (i.e., pruned if its suffix
is not the smallest one labeled $Y$, or expanded otherwise). If we choose
$X \rightarrow Y\cdot Y^{t-1}$, let $Y$ expand to $T[x\dd y-1]$ and $Y^{t-1}$
expand to $T[y\dd z]$. We then define $T[y\dd z]$ as the target of the source
$T[x\dd x+z-y]$. If we instead choose $X \rightarrow Y^{t-1} \cdot Y$, then
we define $T[x\dd x+z-y]$ as the target of the source $T[y\dd z]$. In both
cases, the target overlaps the source if $t>2$.

Note that one of those two cases must copy a source to a larger target, 
depending on whether $T[x\dd] \prec T[y\dd]$ or $T[y\dd] \prec T[x\dd]$. 
Further, the argument used in the proof of Theorem~\ref{thm:ordered-g} to show that the 
copy is valid when the order is extensible, is also valid when source and target
overlap. Thus, we
obtain an ordered parse. Since we have $g_{rl}+1$ leaves in the grammar
tree, the ordered parse is of size $u \le g_{rl}+1$, and therefore the optimal 
one is also of size $u \le g_{rl}+1$. By Theorem~\ref{thm:rlcfg}, we also have 
$u = O(b\log(n/b))$.
\end{proof}

Finally, we show that greedy parsings can be 
computed much faster on extensible orders.

\begin{theorem} \label{thm:nlogn}
Any ordered parse, under any extensible order $\preceq$, can be computed greedily in $O(n)$ expected time or $O(n\log\log\sigma)$ worst-case time, and $O(n)$ space, given an array $O[1\dd n]$ with the suffixes of $T$ sorted by $\preceq$.
\end{theorem}
\begin{proof}
We first compute the suffix array $\SA$ of $T$ in $O(n)$ time (recall Section~\ref{sec:bwt}),
and from it, the {\em suffix tree} of $T$ \cite{Wei73} can also be built in $O(n)$ time \cite{KLAAP01}.
The suffix tree is a compact trie storing all the suffixes of $T$, so that we can descend from the root
and, in time $O(m)$, find the interval $\SA[sp\dd ep]$ of all the suffixes starting with a given string of length $m$. 

We also create in $O(n)$ time the inverse permutation $IO[1\dd n]$ of $O[1\dd n]$, that is, $IO[p]$ is the rank of $T[p\dd]$ among all the other suffixes, in the order $\preceq$. With it, we build in $O(n)$ time a {\em range minimum query} data structure on the virtual array $K[k] = IO[\SA[k]]$, so that $RMQ(i,j) = \arg\min_{i\le k \le j} K[k]$ is computed in constant time \cite{FH11}. Therefore, if $\SA[sp\dd ep]$ is the suffix array interval of all the suffixes $T[p\dd]$ starting with $S$, then $RMQ(sp,ep)$ gives the
suffix starting with $S$ with the minimum rank in the order $\preceq$.


We now create the parse phrase by phrase. To produce the next phrase $T[p\dd]$, we enter the suffix tree from the root with the successive
symbols $T[p+j]$, for $j \ge 0$. At each step, the suffix tree gives us the range $\SA[sp_j,ep_j]$ of the suffixes of $T$ starting with $T[p\dd p+j]$.
We then find $K[RMQ(sp_j,ep_j)]$, which is the smallest rank of any occurrence of $T[p\dd p+j]$ in $T$. If this is less than $IO[p]$, then there is a smaller 
occurrence of $T[p\dd p+j]$ and we 
continue with the next value of $j$. The process stops when $K[RMQ(sp_j,ep_j)] = IO[p]$, that is, $T[p\dd p+j]$ is its smallest occurrence, so we cannot copy it from a smaller one. At this point, the new phrase is
$T[p\dd p+j-1]$ if $j>0$, or the explicit symbol $T[p]$ if $j=0$.


Since we descend to a suffix tree child for every symbol of $T$, the total
traversal time is $O(n)$ as well. There is a caveat, however. To provide
constant-time traversal to children, the suffix tree must implement perfect
hashing on the children of each node, which can be built in constant
expected time per element. In this case, the whole parsing takes $O(n)$
expected time. Alternatively, each node can store its children with a predecessor data structure, so that each traversal to a child costs $O(\log\log\sigma)$ time, and the structure can be built in worst-case time $O(n\log\log\sigma)$ \cite[Sec.~A.1 \& A.2]{BN14}, which dominates the total worst-case time of the parsing. 
The total space used is $O(n)$ in both variants. If array $O[1\dd n]$ is not given, we can compute it with a classical sorting algorithm in $O(n\log n)$ evaluations of $\prec$.
\end{proof}


\section{Lexicographic Parses}

In this section we study a particularly interesting ordered parse we call
lexicographic parse.

\begin{definition}
A \emph{lexicographic parse} of $T[1\dd n]$ is an ordered parse where 
$T[s_i\dd ] \prec T[t_i\dd ]$ iff the former suffix is smaller than the latter
in lexicographic order, or which is the same, if $\ISA[s_i] < \ISA[t_i]$.
\end{definition}

We first note that the order we use is extensible.

\begin{lemma}
The order $T[s\dd] \prec T[t\dd]$ iff the suffix $T[s\dd]$ lexicographically precedes $T[t\dd]$, is extensible.
\end{lemma}
\begin{proof}
If $T[s\dd]$ lexicographically precedes $T[t\dd]$
and $T[s]=T[t]$, then by definition of lexicographic order, $T[s+1\dd]$ lexicographically precedes $T[t+1\dd]$.
\end{proof}

By Lemma~\ref{lem:ordered-to-b}, then, any  lexicographic parse is a bidirectional
scheme. 
One example of a lexicographic parse is the bidirectional scheme based on the 
BWT we introduced in Section~\ref{sec:lower}.

\begin{lemma}\label{lemma:RLBWT parse}
The bidirectional scheme induced by the BWT in Def.~\ref{def:bwtparse}
is a lexicographic parse of size $2r$.
\end{lemma}
\begin{proof}
The definition uses function $f(p) = \phi(p) = \SA[\ISA[p]-1]$ to copy from
$T[\phi(t_i)\dd \phi(t_i)+\ell_i-1]$ to $T[t_i\dd t_i+\ell_i-1]$, where
$\ell_i = t_{i+1}-t_i-1$ (recall Theorem~\ref{thm:r}). Therefore it holds that
$\ISA[s_i] = \ISA[\phi(t_i)] = \ISA[t_i]-1 < \ISA[t_i]$.
\end{proof}

Another lexicographic parse is {\em lcpcomp} \cite{DFKLS17}. This algorithm 
uses a queue to find the largest entry in the $\LCP$ array (recall Section~\ref{sec:bwt}). This information 
is then used to define a new phrase of the factorization. $\LCP$ entries 
covered by the phrase are then removed from the queue, LCP values affected by the creation of the new phrase are decremented,  and the process is 
repeated until there are no text substrings that can be replaced with a pointer to lexicographically smaller positions. The output of lcpcomp is a series of source-length pairs interleaved with plain substrings (that cannot be replaced by pointers). 

\begin{lemma}\label{lemma:lcpcomp}
	The \emph{lcpcomp} factorization \cite{DFKLS17} is a lexicographic parse.
\end{lemma}
\begin{proof}
The property can be easily seen from step 2 of the algorithm~\cite[Sec. 3.2]{DFKLS17}: the authors report a phrase (i.e., source-length pair) $(\SA[i-1], \LCP'[i])$ expanding to text substring $T[\SA[i]\dd \SA[i]+\LCP'[i]-1]$. We write $\LCP'[i]$ because entries of the LCP array may be decremented in step 4, therefore $\LCP'[i]\leq \LCP[i]$ at any step of the algorithm for any $1\leq i\leq n$. This however preserves the two properties of lexicographic parsings: $T[\SA[i]\dd \SA[i]+\LCP'[i]-1] = T[\SA[i-1]\dd \SA[i-1]+\LCP'[i]-1]$ (phrases are equal to their sources) and, clearly, $i-1 < i$ (sources are lexicographically smaller than phrases).
\end{proof}

Since the lexicographic order is extensible, we can find the optimal lexicographic parse  greedily,
in $O(n\log\log \sigma)$ time, by Theorem~\ref{thm:nlogn}. We now show that, just as
Lempel-Ziv, it can be found in $O(n)$ time, in a surprisingly simple way.

\begin{definition} \label{def:lex}
The {\em lex-parse} of $T[1\dd n]$, with arrays $\SA$, $\ISA$, and $\LCP$,
is defined as a partition $T=L_1,\ldots, L_v$ such that 
$L_i=T[t_i\dd t_i+\ell_i-1]$,
satisfying $(1)$ $t_1=1$ and $t_{i+1} = t_i+\ell_i$, and $(2)$
$\ell_i = \LCP[\ISA[t_i]]$, with the exception that if $\ell_i=0$ we set
$\ell_i=1$ and make $L_i$ an explicit symbol. The non-explicit phrases
$T[t_i\dd t_i+\ell_i-1]$ are copied from $T[s_i\dd s_i+\ell_i-1]$, where
$s_i = \SA[\ISA[t_i]-1]$.
\end{definition}

\begin{example}
The lex-parse of our example string is 
$a|\underline{l}|a|\underline{b}|a|\underline{r}|a|\underline{l}|alabar|\underline{d}|\underline{a}|\underline{\$}$, where we underlined the explicit symbols.
\end{example}\medskip

Since $\ISA$ and $\LCP$ can be built in linear time, it is clear that the 
lex-parse of $T$ can be built in $O(n)$ time. Let us show that it is 
indeed a valid lexicographic parse.

\begin{lemma}
  The lex-parse is a lexicographic parse.
\end{lemma}
\begin{proof}
First, the parse covers $T$ and it copies sources to targets with the
same content: Let $x=\ISA[t_i]$ and $y=\ISA[t_i]-1$. Then $\ell_i =
\LCP[x]$ is the length of the shared prefix between $t_i=\SA[x]$ and 
$s_i=\SA[y]$. Therefore we can copy $T[s_i\dd s_i+\ell_i-1]$ to 
$T[t_i\dd t_i+\ell_i-1]$.
Second, the parse is lexicographic: $\ISA[s_i] = \ISA[t_i]-1 < \ISA[t_i]$.
\end{proof}

From now on we will use $v$ as the size of the lex-parse.
Let us show that the lex-parse is indeed ordered-optimal.

\begin{theorem}\label{thm:smallest}
	The lex-parse is the smallest lexicographic parse. Thus, $v \le 2r$,
$v \le |lcpcomp|$, $v = O(b\log(n/b))$, $v\log_2 n \le nH_k + o(n\log\sigma)$ for any $k=o(\log_\sigma n)$, $v = 
O(n/\log_\sigma n)$, and there are text families where $r = \Omega(v\log n)$.
\end{theorem}
\begin{proof}
By Theorem~\ref{thm:greedy}, it suffices to show that Def.~\ref{def:lex}
defines a greedy parse under lexicographic ordering. Indeed, $\ell_i =
\LCP[\ISA[t_i]]$ is the longest prefix shared between $T[t_i\dd]$ and any
other suffix that is lexicographically smaller than it.

The other results are immediate consequences of 
Lemmas~\ref{lemma:RLBWT parse} and \ref{lemma:lcpcomp},
Theorem~\ref{thm:ordered-grl}, 
Lemma~\ref{lem:entropy},
and Theorem~\ref{thm:ordered-g}.
\end{proof}

Note that, unlike $v$, $z$ can be $\Omega(r \log n)$, as shown in 
Theorem~\ref{thm:lowz}. Thus, $v$ offers a better asymptotic bound with 
respect to the number of runs in the BWT. 
The following corollary is immediate.

\begin{theorem} \label{thm:lowzv}
There is an infinite family of strings over an alphabet of size 2 for which
$z = \Omega(v \log n)$.
\end{theorem}

We now show that the bound $v=O(b\log(n/b))$ is tight as a function of $n$. 

\begin{theorem} \label{thm:vb}
  There is an infinite family of strings over an alphabet of size 2
  for which $v = \Omega(b\log n)$.
\end{theorem}
\begin{proof}
We first prove that $b \le 4$ for all Fibonacci words, and
then that $v = \Omega(\log n)$ on the {\em odd} Fibonacci words (on the even
ones it holds $v=O(1)$, by Theorem~\ref{thm:lowz}).
The proof is rather technical, so we defer it to Appendix~\ref{sec:fib}.
\end{proof}

An interesting
remaining question is whether $v$ is always $O(z)$ or there is a string
family where $z=o(v)$. While we have not been able to settle this question, we
can exhibit a string family for which $z < \frac{3}{5}v$.

\begin{lemma} \label{lem:35}
On the alphabet $\{ 1,\ldots,\sigma+1 \}$, where $\sigma$ is not a multiple of 
$3$, consider the string $S_1 = (2\,3\ldots \sigma\, 1)^3$. Then, for 
$i=1,\ldots,\sigma-1$, string $S_{i+1}$ is formed by changing 
$S_i[3\sigma-3i]$ to $\sigma+1$. Our final
text is then $T = S_1 \cdot S_2 \cdots S_\sigma$, of length
$n = 3\sigma^2$. In this family, $z=3\sigma-2$ and $v=5\sigma-2$.
\end{lemma}
\begin{proof}
In the Lempel-Ziv parse of $T$, we first have $\sigma+1$ phrases of length 1 to cover
the first third of $S_1$, and then a phrase that extends in $T$ until the first
edit of $S_2$. Since then, each edit forms two phrases: one covers the edit
itself (since $\sigma$ is not a multiple of $3$, each edit is followed by a 
distinct symbol), and the other covers the range until the next edit. This adds
up to $z = 3\sigma-2$.

A lex-parse starts similarly, since the Lempel-Ziv phrases indeed point to 
lexicographically smaller ones. However, it needs $2\sigma$ further phrases to 
cover $S_\sigma = 2\,3\,(\sigma+1)\,5\,6\,(\sigma+1)\ldots$ 
with phrases of alternating length $2$ and $1$: each such
pair of suffixes $S_\sigma[3i+1\dd ]$ and $S_\sigma[3i+3\dd ]$, for $i=0,\ldots,
\sigma-1$, do appear in previous substrings $S_j$, but all these are 
lexicographically larger (because $\sigma$ is not a multiple of $3$, and
thus symbols $1$ are never replaced by $\sigma+1$). Therefore, 
only length-2 strings of symbols not including $\sigma+1$ can point to,
say, $S_1$ (this reasoning has been verified computationally as 
well). This makes a total of $v=5\sigma-2$ phrases.
\end{proof}

\subsection{Experimental Comparison with Lempel-Ziv}

As a test on the practical relevance of the lex-parse, we measured $v$, $z$, and $r$
on various synthetic, pseudo-real, and real repetitive collections obtained
from PizzaChili ({\tt http://pizzachili.dcc.uchile.cl}) and on four repetitive collections (\texttt{boost}, \texttt{bwa}, \texttt{samtools}, \texttt{sdsl}) obtained by concatenating the first versions of  github repositories ({\tt https://github.com}) until obtaining a length of $5\cdot 10^8$ characters for each collection.

Table~\ref{tab:lex} shows the results. Our new lex-parse performs better than 
Lempel-Ziv on the
synthetic texts, especially on the Fibonacci words (\texttt{fib41}), the family for which we
know that $v=o(z)$ (recall Theorems~\ref{thm:lowz} and \ref{thm:lowzv}).\footnote{The file \texttt{fib41} uses a variant where $F_1=a$, $F_2=ba$, and $F_k = F_{k-2} F_{k-1}$.} 
On the others (Run-Rich String and Thue-Morse sequences), $z$ is about 30\% larger than $v$.

Pseudo-real texts are formed by taking a real text and replicating it many
times; a few random edits are then applied to the copies. The fraction of 
edits is indicated after the file name, for example, {\tt sources.001} indicates a
probability of 0.001 of applying an edit at each position. In the names with 
suffix {\tt .1}, the edits are applied to the base version to form the copy,
whereas in those with suffix {\tt .2}, the edits are cumulatively applied to the
previous copy. It is interesting to note that, in this family, $v$ and $z$ are
very close under the model of edits applied to the base copy, but $z$ is 
generally significantly smaller when the edits are cumulative. The ratios
actually approach the $\frac{3}{5} = 0.6$ we obtained in Lemma~\ref{lem:35}
using a particular text that, incidentally, follows the model of cumulative
edits.

On real texts, both measures are very close. Still, it can be seen that in
collections like {\tt einstein.de} and {\tt einstein.en}, which feature 
cumulative edits (those collections are formed by versions of the Wikipedia
page on Einstein in German and English, respectively), $z$ is about 8\% 
smaller than $v$. On the other hand, $v$ is about 3\%--4\% smaller than $z$ on 
biological datasets such as {\tt cere}, {\tt escherichia\_coli} and {\tt para},
where the model is closer to random edits applied to a base text. The lex-parse
is also about 1\% smaller than the Lempel-Ziv parse on github versioned collections,
except \texttt{bwa}.

To conclude, the comparison between $r$ and $v$ shows that the sub-optimal 
lexicographic parse induced by the Burrows-Wheeler transform is often much 
larger (typically 2.5--4.0 times, but more than 7 times on the biological datasets) than the optimal 
lex-parse. Interestingly, on Fibonacci words the optimal parse is already 
found by the Burrows-Wheeler transform.

\begin{table}[tp]
	\centering
	\begin{tabular}{|@{~}l@{~}|@{~}r@{~}|@{~}r@{~}|@{~}r@{~}|@{~}r@{~}|@{~}c@{~}|@{~}c@{~}|}\hline
file		&$n~~~~~~~$	&$r~~~~~~$	&$z~~~~~$	&$v~~~~~$ 	& $z/v~$ & $r/v$ \\
\hline
fib41	&267,914,296	&4	&41	&4	&$>10$ & $1.000$ \\
rs.13	&216,747,218	&77	&52	&40	&$1.300$ & $1.925$ \\
tm29	&268,435,456	&82	&56	&43	&$1.302$ & $1.907$ \\
\hline
dblp.xml.00001.1	&104,857,600	&172,489	&59,573	&59,821	&$0.996$ & $2.883$ \\
dblp.xml.00001.2	&104,857,600	&175,617	&59,556	&61,580	&$0.967$ & $2.852$ \\
dblp.xml.0001.1	&104,857,600	&240,535	&78,167	&83,963	&$0.931$ & $2.865$ \\
dblp.xml.0001.2	&104,857,600	&270,205	&78,158	&100,605	&$0.777$ & $2.686$ \\
sources.001.2	&104,857,600	&1,213,428	&294,994	&466,643	&$0.632$ & $2.600$\\
dna.001.1	&104,857,600	&1,716,808	&308,355	&307,329	&$1.003$ & $5.586$\\
proteins.001.1	&104,857,600	&1,278,201	&355,268	&364,093	&$0.976$ & $3.511$ \\
english.001.2	&104,857,600	&1,449,519	&335,815	&489,034	&$0.687$ & $2.964$\\
\hline
boost	&500,000,000	&61,814	&22,680	&22,418	&$1.012$ & $2.757$ \\
einstein.de	&92,758,441	&101,370	&34,572	&37,721	&$0.917$ & $2.687$\\
einstein.en	&467,626,544	&290,239	&89,467	&97,442	&$0.918$ & $2.979$ \\
bwa	&438,698,066	&311,427	&106,655	&107,117	&$0.996$ & $2.907$ \\
sdsl	&500,000,000	&345,325	&113,591	&112,832	&$1.007$ & $3.061$ \\
samtools	&500,000,000	&458,965	&150,988	&150,322	&$1.004$ & $3.053$ \\
world\_leaders	&46,968,181	&573,487	&175,740	&179,696	&$0.978$ & $3.191$ \\
influenza	&154,808,555	&3,022,822	&769,286	&768,623	&$1.001$ & $3.933$ \\
kernel	&257,961,616	&2,791,368	&793,915	&794,058	&$1.000$ & $3.515$ \\
cere	&461,286,644	&11,574,641	&1,700,630	&1,649,448	&$1.031$ & $7.017$ \\
coreutils	&205,281,778	&4,684,460	&1,446,468	&1,439,918	&$1.005$ & $3.253$ \\
escherichia\_coli	&112,689,515	&15,044,487	&2,078,512	&2,014,012	&$1.032$ & $7.470$ \\
para	&429,265,758	&15,636,740	&2,332,657	&2,238,362	&$1.042$ & $6.986$ \\
\hline
\end{tabular}

\vspace{5mm}

\caption{Various repetitiveness measures obtained from synthetic, pseudo-real,
and real texts (each category forms a block in the table).}
\label{tab:lex}
\end{table}

\section{Bounds on Collage Systems}

In this
section we use our previous findings to prove that $c = O(z)$, $b = O(c)$, 
and that there exist string families where $c = \Omega(b\log n)$, where $c$ is the size of the smallest (internal) collage system.

\begin{theorem} \label{thm:ztoc}
There is always an internal collage system of $c \le 4z$ rules generating $T$. 
\end{theorem}
\begin{proof}
We proceed by induction on the Lempel-Ziv parse. At step $i$, we obtain a 
collage system with initial rule $S_i$ that generates the prefix $T[1\dd p_i]$ 
of $T$ covered by the first $i$ phrases. The initial symbol for the whole $T$
is then $S_z$.

For the first phrase, which must be an explicit symbol $a$, we insert the rule
$S_1 \rightarrow a$. Let us now consider the phrases $i>1$.
If the $i$th phrase is an explicit symbol $a$, then we add rules $A_i
\rightarrow a$ and $S_i \rightarrow S_{i-1} A_i$. 

Otherwise, let the $i$th
phrase point to a source that is completely inside $T[1\dd p_{i-1}]$, precisely 
$T[x\dd y]$ with $y \le p_{i-1}$. Then we add rule $N_i \rightarrow 
S_{i-1}^{[x,y]}$, and then $S_i 
\rightarrow S_{i-1} N_i$. 

If, instead, the $i$th phrase points to a source that overlaps it, 
$T[x\dd y]$ with $p_{i-1} < y < p_i$, then $T[x\dd y]$ is periodic
with period $p = p_{i-1}-x+1$, that is, $T[x\dd y-p] = 
T[x+p\dd y]$. Therefore, the new phrase is formed by $q=\lfloor\frac{y-x+1}{p}\rfloor$ copies of $T[x\dd x+p-1] = 
T[x\dd p_{i-1}]$ plus $T[x\dd x+((y-x+1)\!\!\mod p)-1]$ if $p$ does not divide
$y-x+1$ (note that $q$ may be zero). This can be obtained with $O_i \rightarrow ^{[p]}\!\!S_{i-1}$, $O_i' \rightarrow O_i^{[(y-x+1) \bmod p]}$,
$R_i \rightarrow O_i^q$, $N_i \rightarrow R_i O_i'$,
and $S_i \rightarrow S_{i-1}N_i$.

Figure~\ref{fig:ztoc} illustrates both cases schematically.
\end{proof}

\begin{figure}[t]
\includegraphics[width=\textwidth]{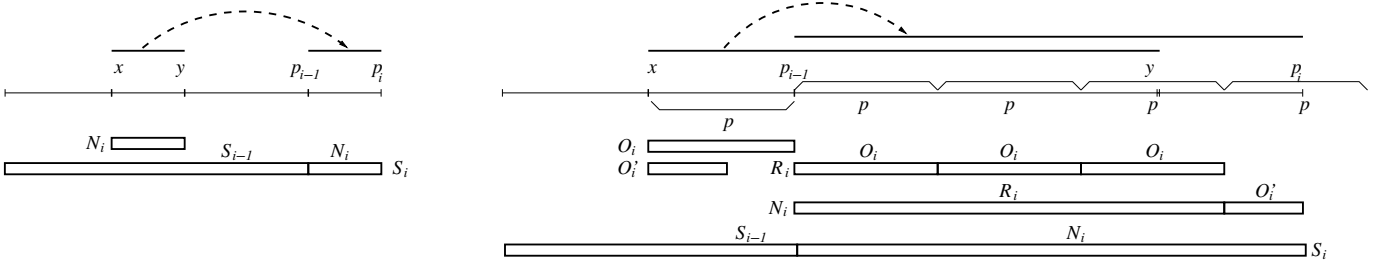}
\caption{Conversion of a Lempel-Ziv parse into a collage system using Theorem~\ref{thm:ztoc}. On the left, the nonoverlapping case. On the right, the overlapping case.}
\label{fig:ztoc}
\end{figure}

\begin{example}
Consider the Lempel-Ziv parse $T = \underline{a}|\underline{l}|a|\underline{b}|a|\underline{r}|ala|labar|\underline{d}|a|\underline{\$}$ of Section~\ref{sec:lz}, where we have underlined the explicit symbols. Figure~\ref{fig:zc-ex} illustrates the application of Theorem~\ref{thm:ztoc} to this parse.
\end{example}

\begin{figure}[t]
\begin{center}
\includegraphics[width=0.4\textwidth]{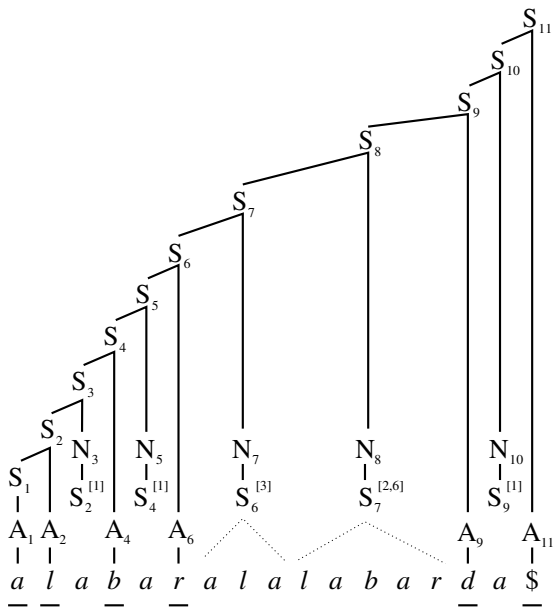}
\end{center}
\caption{Creation of an internal collage system from the Lempel-Ziv parse of $T = \underline{a}|\underline{l}|a|\underline{b}|a|\underline{r}|ala|labar|\underline{d}|a|\underline{\$}$, using Theorem~\ref{thm:ztoc}.}
\label{fig:zc-ex}
\end{figure}

\begin{theorem}
There is always a bidirectional scheme of size $b \le c+1$ for $T$, for an internal collage system of size $c$.
\end{theorem}
\begin{proof}
We extend the idea of  Theorem~\ref{thm:grlz} to handle substring rules. We draw the parse tree of $T$, starting from the initial symbol. When we reach a nonterminal defined by a substring rule, we convert it into a leaf. Just as for grammar trees, we also convert into leaves all but the leftmost occurrence of each other nonterminal in the parse tree. Analogously to grammar trees, the resulting tree has at most $c+1$ leaves, because we are just adding substring rules, each of which adds a new leaf. 

We now generate a bidirectional macro scheme exactly as we defined the left-to-right parse in Theorem~\ref{thm:grlz}.
Further, each leaf representing a substring rule $A \rightarrow B^{[t,t']}$ is converted into a single phrase pointing to $T[x+t-1\dd x+t'-1]$, where the leftmost occurrence of $B$ in the parse tree covers the text $T[x\dd y]$.

The resulting parse may not be left-to-right anymore. However, it is a valid bidirectional scheme. To see this, let us label each position $p$ in $T$ with the index in the sequence of rules of the leaf of the grammar tree covering $T[p]$. This means that the labels of text positions descending from an internal node $A$ are smaller than the index of $A$.  Since nonterminals are defined in terms of earlier nonterminals, it holds that every position $p$ of $T$ is defined in terms of a position $f(p)$ with a smaller label.
\end{proof}

\begin{example}
The following collage system to generate the text $T=alabaralalabarda\$$
is an internal variant of the one given in Section~\ref{sec:collage}:
$A \rightarrow a$, $B \rightarrow b$, $D \rightarrow d$, $L \rightarrow l$, $R \rightarrow r$, 
$Z \rightarrow \$$, $C \rightarrow AL$, $E \rightarrow CC$, $F \rightarrow BA$, $G \rightarrow FR$, 
$H \rightarrow DA$, $I \rightarrow HZ$, $J \rightarrow EA$, $K \rightarrow JG$,
$M \rightarrow \,^{[6]} \!K$, $N \rightarrow MK$, $O \rightarrow NI$.
The corresponding bidirectional scheme induces the parse $T=alabar|\underline{a}|\underline{l}|al|a|\underline{b}|a|\underline{r}|\underline{d}|a|\underline{\$}$, where the first phrase is defined by a forward pointer to $T[9\dd 14]$.
\end{example}\medskip

\begin{theorem}
There exists an infinite family of strings over an alphabet of size 2 for which
$c = \Omega(v \log n)$, and thus also $c = \Omega(b\log n)$, for any general collage system of size $c$.
\end{theorem}
\begin{proof}
Fibonacci words do not contain 4 consecutive repetitions of the same substring 
\cite{Kah83}. Therefore, no internal collage system generating a Fibonacci word contains
run-length rules $A \rightarrow B^k$ with $k>3$, because $exp(A)$ does appear in $T$. Run-length rules with $k \le 3$
can be replaced by one or two rules that are not run-length rules. Therefore,
if a Fibonacci word of length $n$ is generated by an internal collage system
of size $c$, then it is also generated by an internal collage system of size at most $2c$
with no run-length rules.

Just as with SLPs, no such collage system can generate a string of 
length more than $2^{2c}$; the substring rules do not help in obtaining
strings of some length with fewer rules. As a consequence, it holds that
$c = \Omega(\log n)$. On the other hand, by Theorem~\ref{lem:bms},
it holds that Fibonacci words have bidirectional schemes of $O(1)$ blocks.
Further, by Theorems~\ref{thm:lowz} and  \ref{thm:smallest},
it holds that $v=O(1)$ on the even Fibonacci words.

We can extend the result to general collage systems by noting that every nonterminal $A \rightarrow B^k$ with $k>4$ must be shortened via truncation by more than $|exp(B)|$ symbols, before appearing in $T$. Thus, it can be replaced by $A \rightarrow B^{k-1}$ and, iteratively, by $A \rightarrow B^4$, and thus be replaced by two rules that are not run-length rules.
\end{proof}


\section{Conclusions}

We have essentially closed the question of which the approximation ratio of
the (unidirectional, left-to-right) Lempel-Ziv parse is with respect to the 
optimal bidirectional
parse, therefore contributing to the understanding of the quality of this
popular heuristic that can be computed in linear time, whereas computing the
optimal bidirectional parse is NP-complete. Our bounds, which are shown to be
tight, imply that the gap is in fact logarithmic, wider than what was 
previously known.

We have then generalized Lempel-Ziv to the class of optimal ordered parsings,
where there must be an increasing relation between source and target positions in a copy. We
proved that some features of Lempel-Ziv, such as converging to the empirical
entropy, being limited by the smallest RLSLP, and being worse than the optimal bidirectional scheme by at most a
logarithmic factor, hold in fact for all 
optimal ordered parsings.

\begin{figure}[t]
\centerline{\includegraphics[width=8.5cm]{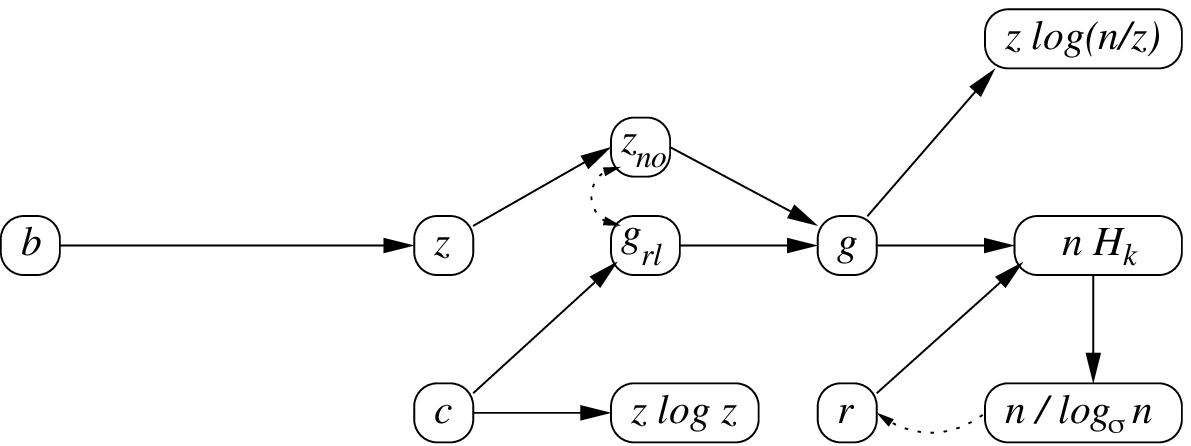}
\hfill
\includegraphics[width=8.5cm]{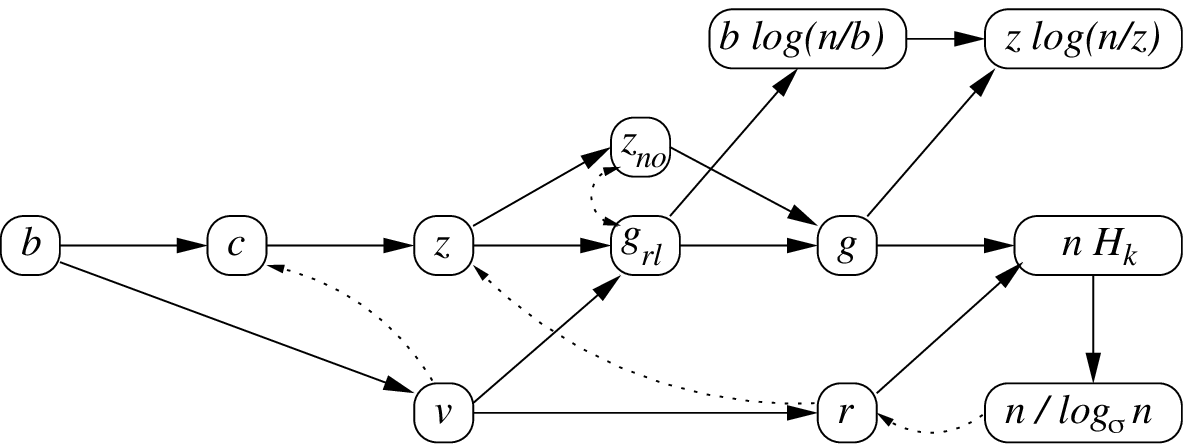}}
\caption{Previously known (left) and new (right) asymptotic bounds between 
repetitiveness measures. An arrow from $x$ to $y$ means that $x = O(y)$ for
every string family. The arrow $b \rightarrow c$ holds for internal collage systems only. For most arrows, a logarithmic gap for some string
family is known, except $c \rightarrow z$. There are also logarithmic gaps for 
some incomparable measures, shown in dotted lines (one is less than
logarithmic, $g_{rl}=\Omega(z_{no}\log n/\log\log n)$).}
\label{fig:bounds}
\end{figure}

As an example of such a parse, we introduced the lex-parse, which is the optimal 
left-to-right parse in the lexicographical order of the involved suffixes. This 
new parse is shown to be computable greedily in linear time and to have many of
the good bounds of the Lempel-Ziv parse with respect to other measures, even 
improving on some. For example, being an optimal ordered parse, the lex-parse is upper-bounded by the smallest
RLCFG and it is an approximation to the smallest bidirectional parse with a logarithmic gap. In addition, the lex-parse is bounded by the number of runs in the BWT of the text, which is not the case of the Lempel-Ziv parse.
We exhibit a family of strings where the lex-parse is 
asymptotically smaller than the Lempel-Ziv parse, and another where the latter is
smaller than the lex-parse, though only by a constant factor. Experimentally,
the lex-parse is shown to behave similarly to the Lempel-Ziv parse, although it is 
somewhat larger on versioned collections with cumulative edits.

Finally, we showed that the smallest collage systems are of the order of the Lempel-Ziv parse. A restricted variant we call internal collage systems are shown not to be asymptotically smaller than the smallest bidirectional scheme, and have a logarithmic gap
with the lex-parse on some string families. Many other results are proved along the way.

Figure~\ref{fig:bounds} illustrates the contributions of this article to the
knowledge of the asymptotic bounds between repetitiveness measures.
%
Note that the solid arrow relations are transitive, because they hold for every
string family. Dotted arrows, instead, are not transitive because they hold
for specific string families. 

There are various interesting avenues of future work. For example, it is unknown
if there are string families where $z=o(v)$ or $c=o(z)$, nor if $b=O(c)$ holds for general collage systems. We can prove the latter if it holds that $b$ grows only by a constant factor when we remove a prefix of $T$, but this is an open question.
We can even prove $z=O(c)$ for general collage systems if it holds that there is only a constant gap between $z$ for $T$ and for its reverse, which is another open question. We have also
no upper bounds on $r$ in terms of other measures, for example, can 
$r$ be more than $O(\log n)$ times larger than $z$ or $g$? It might also
be that our Theorem~\ref{thm:rlcfg} can be proved without using run-length 
rules, then yielding $g = O(b\log(n/b))$.  

Another interesting line of work is that of optimal ordered parses,
which can be built efficiently and compete with $z$, which has been the 
gold-standard approximation for decades. Are there other convenient parses
apart from our lex-parse? In particular, are there parses that can compete 
with $z$ while offering efficient random access time to $T$? Right now, only 
parses of size $O(g)$ (and $O(g_{rl})$ \cite{CEKNP19}) allow for efficient 
($O(\log n)$ time) access to $T$; all the other measures need a logarithmic
blowup in space to support efficient access 
\cite{BGGKOPT15,BLRSRW15,BPT15,SG06,GGKNP12,gagie2014lz77,GNP18soda}. 
This is also crucial to build small and efficient compressed indexes on
$T$ \cite[Sec.~13.2]{Nav16}.

\section*{Acknowledgements}
This work was partially funded by Basal Funds FB0001, Conicyt, by Fondecyt
Grant 1-170048, Chile, by the Millennium Institute for Foundational Research on Data (IMFD), by the Danish Research Council Fund DFF-4005-00267, and by the project MIUR-SIR CMACBioSeq, grant n.~RBSI146R5L.

\bibliographystyle{plain}
\bibliography{paper}

\appendices

\section{A separation between $b$ and $v$} \label{sec:fib}

In this section we prove  that $b \le 4$ for all Fibonacci words, and then that $v= \Omega(\log n)$ on the odd Fibonacci words. We first state a couple of results on Fibonacci words $F_k$.

\begin{lemma}
  \label{lem:near}
  For each $k \ge 5$, it holds that
  $F_{k-1}F_{k-2} = H_k ba$ and $F_{k-2}F_{k-1} = H_k ab$ if $k$ is even, and
  $F_{k-1}F_{k-2} = H_k ab$ and $F_{k-2}F_{k-1} = H_k ba$ if $k$ is odd.
 Note that $|H_k| = f_k-2$.
\end{lemma}
\begin{proof}
It is easy to see by induction that $F_k = F_{k-1}F_{k-2}$ finishes with $ab$ 
if $k$ is odd and with $ba$ if $k$ is even. The fact that $F_{k-1}F_{k-2} = 
H_k xy$ and $F_{k-2}F_{k-1} = H_k yx$ was proved by Pirillo \cite[Lem.~1]{P97}.
\end{proof}

\begin{lemma}
  \label{lem:noinside}
$F_{k-1}$ only appears at position $1$ in $F_k$.
\end{lemma}
\begin{proof}
Consider the following derivation (which is also used later), obtained by
applying Def.~\ref{def:fibonacci} several times:
\begin{eqnarray} 
F_k & = & F_{k-1}F_{k-2} \nonumber \\
    & = & F_{k-2}F_{k-3}F_{k-2} \label{eq1} \\
    & = & F_{k-2}F_{k-3}F_{k-3}F_{k-4} \label{eq3} \\
    & = & F_{k-2}F_{k-3}F_{k-4}F_{k-5}F_{k-4} \nonumber \\
    & = & F_{k-2}F_{k-2}F_{k-5}F_{k-4}. \label{eq2}
\end{eqnarray}

  Assume, by contradiction,
  that $F_{k-1}$ appears in two different positions
  inside $F_k$.
  From Eq.~(\ref{eq2}), we have that
  $F_n = F_{k-2}F_{k-2}F_{k-5}F_{k-4}$.
  Also, no occurrence of $F_{k-1}$ can start after
  position $f_{k-2}$ in $F_k$ (because it would exceed $F_k$ unless it starts at
$p=f_{k-2}+1$, but this is also outruled because 
$F_k = F_{k-1}F_{k-2} \not= F_{k-2}F_{k-1}$ by Lemma~\ref{lem:near}).
  Thus, the second occurrence of $F_{k-1}$ must
  start at a position $p \le f_{k-2}$.
  Then, by Eq.~(\ref{eq2}) again, there is a third occurrence of $F_{k-2}$ 
  within $F_{k-2}F_{k-2}$, which means that
  $F_{k-2}$ appears twice in the circular rotations of $F_{k-2}$.
  Yet, this is a contradiction because 
  all the circular rotations on the Fibonacci words are
  different~\cite[Cor.~3.2]{D95}.
\end{proof}

\begin{lemma}
  \label{lem:bms}
  Every word $F_k$ has a bidirectional scheme of size $b \le 4$.
\end{lemma}

\begin{proof}
  Up to $k=4$ we have $|F_k| \le 3$, so the claim is trivial. For $F_5 =
abaab$ we can copy the last $ab$ from the first to have $b=4$. For $k \ge 6$,
consider the following partition of $F_k = F_{k-1}F_{k-2}$ into $4$ chunks:
  \begin{enumerate}
  \item The first chunk is $B_1 = F_k[1\dd f_{k-1}-2]$
    (i.e., all the symbols of $F_{k-1}$ except the last two).
  \item The second and third chunks are explicit symbols  
    ($B_2 = F_k[f_{k-1}-1] = b$ and $B_3 = F_k[f_{k-1}] = a$,
    if $k$ is even, and
    $B_2 = F_k[f_{k-1}-1] = a$ and $B_3 = F_k[f_{k-1}] = b$,
    if $k$ is odd).
  \item The fourth chunk is $B_4 = F_k[f_{k-1}+1\dd f_k]$
    (i.e., all the symbols of $F_{k-2}$).
  \end{enumerate}
  
  The source of the first chunk, $B_1$, is $F_k[f_{k-2}+1\dd f_k-2]$, and
  the source of the fourth chunk, $B_4$, is $F_k[f_{k-2}+1\dd 2f_{k-2}]$.
  Note that the sources of $B_1$ and $B_4$ start at the same
  position.
  We now prove that this is a valid bidirectional scheme.
  
  First, we prove that $B_1$ and $B_4$ are equal to their sources.
  By Eq.~(\ref{eq2}), $F_k = F_{k-2}F_{k-2}F_{k-5}F_{k-4}$, so there is an 
occurrence of $F_{k-2}$ starting at position $f_{k-2}+1$ of $F_k$.
  Hence, $B_4 = F_k[f_{k-2}+1\dd 2f_{k-2}]$.
  Further, by Eq.~(\ref{eq1}), we have that $F_k = F_{k-2}F_{k-3}F_{k-2}$, 
and from Lemma~\ref{lem:near}
  we have that $B_1= H_{k-1} = F_k[f_{k-2}+1\dd f_k-2]$.

  Thus, the sources of $B_1$ and $B_4$ are correctly defined. We now prove there
are no cycles. Our bidirectional scheme defines the
  function $f: [1\dd f_k] \rightarrow [1\dd f_k] \cup \{-1\}$ as follows:

  \[
    f(p) =
    \begin{cases}
      -1, &\quad\text{if } p = f_{k-1}-1 \text{ or }  p = f_{k-1}  \\
      p + f_{k-2}, &\quad\text{if } p < f_{k-1}-1 \\
      p - f_{k-3}, &\quad\text{if } p > f_{k-1}
    \end{cases}
  \]

  Assume that $f$ has cycles and that a shortest one starts at position $p$.
  Successive applications of $f$ either increase the current position by
  $f_{k-2}$ or decrease the current position by $f_{k-3}$. So, a cycle
  starting at position  $p$ means that $p + xf_{k-2} - yf_{k-3} = p$,
  where $x + y$ is the number of times $f$ was applied; note $x,y>0$
  This is equivalent to $xf_{k-2} = yf_{k-3}$.
  Since $f_{k-2}$ and $f_{k-3}$ are coprime\footnote{Applying Euclid's 
algorithm, we have $\gcd(f_{k-2},f_{k-3})=\gcd(f_{k-3},f_{k-2}-f_{k-3}) =
\gcd(f_{k-3},f_{k-4})$, which is traced down to $\gcd(f_2,f_1)=1$.},
  $f_{k-3}$ divides $x$ and $f_{k-2}$ divides $y$.
  Thus, $x \ge f_{k-3}$, $y \ge f_{k-2}$, and $x+y \ge f_{k-1}$.
  The number of positions involved in a cycle is then at least $f_{k-1}$, and
they must all be different because the cycle is minimal. 
Yet, the first $f_{k-2}$ positions of $F_k$
  cannot be involved in any cycle: once $f$ is applied in one
  of the first $f_{k-2}$ positions there is not way to get back there.
  So, we are left with $f_{k-1}-2$ positions to be involved in a
  cycle, because $f(f_{k-1}-1) = f(f_{k-1}) = -1$.
  That is a contradiction. 
\end{proof}

Before delving into the proof of the lower bound that relates
$v$ and $b$, we prove two further properties of the Fibonacci words we make 
use of.

\begin{lemma}
\label{lem:bbaaa}
The strings $bb$, $aaa$, and $ababab$ never occur within a Fibonacci word.
\end{lemma}

\begin{proof}
It is easy to see that all $F_k$, for $k \ge 3$, start with $ab$. Further, by
Lemma~\ref{lem:near}, they end with $ab$ or $ba$. Then the lemma for $bb$ and
$aaa$ easily follows by induction because, when concatenating $F_k = F_{k-1} 
F_{k-2}$, the new substrings of length 3
we create are substrings of $abab$ or $baab$.
For the third string we easily see that, for $k \ge 5$, every $F_k$ starts
with $abaa$ and ends with $baab$ (odd $k$) or $baba$ (even $k$). Thus, as 
before, it is impossible to form $ababab$ when concatenating any $F_{k-1}$ 
with $F_{k-2}$.
\end{proof}

\begin{lemma}
  \label{lem:sub}
  Given a Fibonacci word $F_k$, for all $4 \le i \le k$,
  every factor $W_i$ of $F_k$ of length
  $f_i$ that begins with $F_{i-1}$ has only
  two possible forms, $W_i = F_{i-1}F_{i-2}$ or $W_i = F_{i-2}F_{i-1}$.
\end{lemma}

\begin{proof}
  We use strong induction on $i$. For the base cases $i=4$ and $i=5$,
  we use the substrings $bb$ and $aaa$ excluded by Lemma~\ref{lem:bbaaa}:
  If $i=4$, then $f_4 = 3$, and $F_3 = ab$. Then, any factor $W_4$ of $F_k$
  of length 3 that begins with $ab$ can only be $W_4 = aba = F_3F_2$.
  If $i=5$, then $f_5 = 5$, and $F_4 = aba$. Then, any factor $W_5$ of $F_k$
  of length 5 that begins with $aba$ can only be equal to $W_5 = abaab = F_4F_3$
  or $W_5 = ababa = F_3F_4$.

  Assume now by induction that, for all $i \ge 4$, every factor 
  $W_i$ of $F_k$ of length
  $f_i$ that begins with $F_{i-1}$ has only two possible forms,
  $W_i = F_{i-1}F_{i-2}$ or $W_i = F_{i-2}F_{i-1}$. We now prove that
  every factor $W_{i+1}$ of $F_k$, of length $f_{i+1}$ and beginning
  with $F_i$, has only two possible forms, $W_{i+1} = F_iF_{i-1}$ or
  $W_{i+1} = F_{i-1}F_i$.

  The factor $W_{i+1}$ is equal to $F_iG_{i-1}$, where $G_x$ will stand for
any string of length $f_x$. 
  Thus, $W_{i+1} = F_{i-1}F_{i-2}G_{i-1}$.
  Since $|F_{i-2}G_{i-1}| = f_i > f_{i-1}$, we can apply the induction 
hypothesis to the first $f_{i-1}$ symbols of this substring.
  Two outcomes are then possible: (i) $W_{i+1} = F_{i-1}F_{i-2}F_{i-3}G_{i-2}$
  or (ii) $W_{i+1} = F_{i-1}F_{i-3}F_{i-2}G_{i-2}$.

  Case (i) implies $W_{i+1} = F_{i-1}F_{i-1}G_{i-2}$.
  By the induction hypothesis, $F_{i-1}G_{i-2} = F_{i-1}F_{i-2}$ or
  $F_{i-1}G_{i-2} = F_{i-2}F_{i-1}$.
  This implies $W_{i+1} = F_{i-1}F_i$ or $W_{i+1} = F_iF_{i-1}$. Thus, $W_{i+1}$
  has the desired form.

  In case  (ii), the suffix $F_{i-2}G_{i-2}$ of $W_{i+1}$ has length over 
$f_{i-1}$ and starts with $F_{i-2}$, so we can apply the induction hypothesis
to obtain subcases (a) $W_{i+1} = F_{i-1}F_{i-3}F_{i-2}F_{i-3}G_{i-4}$
  or (b) $W_{i+1} = F_{i-1}F_{i-3}F_{i-3}F_{i-2}G_{i-4}$.
  We now show that neither subcase is possible.
  In case (a), by Def.~\ref{def:fibonacci}, it holds that
  \begin{equation*}
    \begin{split}
      W_{i+1} & = F_{i-1}F_{i-3}F_{i-2}F_{i-3}G_{i-4} \\
      & = F_{i-2}F_{i-3}F_{i-3}F_{i-2}F_{i-3}G_{i-4} \\
      & = F_{i-2}F_{i-3}F_{i-3}F_{i-3}F_{i-4}F_{i-3}G_{i-4}.
    \end{split}
  \end{equation*}
  If $i+1=6$ or $7$, then $F_{i-3} = a$ or $ab$, and there would be
  3 consecutive occurrences of $a$ or $ab$ in $F_k$, contradicting
  Lemma~\ref{lem:bbaaa}.
  If $i+1 \ge 8$, then by Lemma~\ref{lem:near},
  $F_{i-4}F_{i-3}$ begins with $F_{i-3}$, and
  then there would be 4 consecutive occurrences of
  $F_{i-3}$ within $F_k$, contradicting the fact that
  Fibonacci words do not contain 4 consecutive repetitions
  of the same substring~\cite{Kah83}.
  In case (b), by Def.~\ref{def:fibonacci}, it holds that
  \begin{equation*}
    \begin{split}
      W_{i+1} & = F_{i-1}F_{i-3}F_{i-3}F_{i-2}G_{i-4} \\
      & = F_{i-2}F_{i-3}F_{i-3}F_{i-3}F_{i-2}G_{i-4} \\
      & = F_{i-2}F_{i-3}F_{i-3}F_{i-3}F_{i-3}F_{i-4}G_{i-4},
    \end{split}
  \end{equation*}
which also contains 4 occurrences of $F_{i-3}$ within $F_k$,
a contradiction again~\cite{Kah83}.
\end{proof}

\medskip\noindent{\bf Theorem \ref{thm:vb}.} {\em There is an infinity family of strings over an alphabet of size 2 for which $v = \Omega(b\log n)$.}

\begin{proof} 
  Such a family is formed by the \emph{odd} Fibonacci words, where $b=O(1)$ by
Lemma~\ref{lem:bms}.
  Specifically, we prove that the number of phrases in the lex-parse
  of the odd Fibonacci words forms an arithmetic progression with step 1.

  Let $F_k$ be an odd Fibonacci word with $k \ge 9$.
  We first prove that the length $\ell_1 = \LCP[\ISA[1]]$
  (see Def.~\ref{def:lex}) of the first phrase of
  the lex-parse of $F_k$ is $f_{k-1}-2$.
  From Eq.~(\ref{eq1}), we have that $F_k = F_{k-2}F_{k-3}F_{k-2}$,
  and from Lemma~\ref{lem:near}, we have that
  $F_k = H_{k-1}baF_{k-2} = F_{k-2}H_{k-1}ab$.
  Additionally, $H_{k-1}ab$ is lexicographically smaller than
  $H_{k-1}ba$ and they have a common prefix of length $f_{k-1}-2$.
  Thus, $\ell_1 \ge f_{k-1}-2$.
  We prove that there are no common prefixes of length
  greater than $f_{k-1}-2$ between $F_k$ and any of its suffixes.
Assume the prefix $P_{k-1}$ of length $f_{k-1}-1$ of $F_{k-1}$ appears in $F_k$.
By the proof of Lemma~\ref{lem:bbaaa}, $F_k$ finishes with $baab$ and $F_{k-1}$
finishes with $baba$. Then $P_{k-1}$ finishes with $bab$ and $F_k$ finishes with
$aab$, so $P_{k-1}$ is not a suffix of $F_k$. Also, $b$ can only be followed by
$a$ within $F_k$, by Lemma~\ref{lem:bbaaa}. Hence, if there is an occurrence of
$P_{k-1}$ within $F_k$, then there is also an occurrence of $F_{k-1}$. Yet, the
only occurrence of $F_{k-1}$ in $F_k$ is at the beginning, by
Lemma~\ref{lem:noinside}. Therefore, it is also impossible to find an occurrence
of length $f_{k-1}$ or more.
  
  Next, we prove that the length $\ell_2 = \LCP[\ISA[f_{k-1}-1]]$
  of the second phrase of the lex-parse of $F_k$ is $f_{k-4}+2$.
  By Eq.~(\ref{eq1}), we have that $F_{k-2} = F_{k-4}F_{k-5}F_{k-4}$.
  Since $F_{k-5}$ finishes with $ba$, $baF_{k-4}$ is a prefix and a suffix of 
  $baF_{k-2}$.
Since the suffix is followed by \$, it is lexicographically
smaller than the prefix. Further, since the second phrase starts with the prefix
$baF_{k-4}$, we have $\ell_2 \ge f_{k-4}+2$. We now show that the second phrase
is not longer.

  By the characterization of the
  Fibonacci words of Mantaci et al.~\cite[Thm.~6]{MRS07},
  and the ordering of the cyclic rotations of the
  Fibonacci words stated in there~\cite[proof of Thm.~9]{MRS07},
  the lexicographically smallest cyclic rotation of $F_k$
  is the one that starts at position $x+1$, where
  $x < f_{k}$ is the unique solution to the congruence equation
  $f_{k-2} - 1 + xf_{k-2} \equiv 0 ~(\bmod f_k)$\footnote{Using the notation of Lemma~\ref{lem:smallest}, $R_{f_{k-2}-1}$ is the odd Fibonacci word $F_k$ of length $f_{k-1}+f_{k-2}$, and $R_0$ is the smallest cyclic rotation of $F_k$. 
  Thus, after $x$ applications of $\varrho$ starting at $f_{k-2}-1$,  
  we get the first symbol of $R_0$
  from the first symbol of $F_k$ (i.e., $\varrho^x(f_{k-2}-1) = 0$).}.
  Using Cassini's identity,
  $f_kf_{k-2} - f^2_{k-1} = 1$~\cite{GKP94},
we replace $f_k = f_{k-1}+f_{k-2}$ to get
$f_{k-1} f_{k-2} + f_{k-2}^2 - f_{k-1}^2 =
 f_{k-1} f_{k-2} + (f_{k-2}+f_{k-1})(f_{k-2}-f_{k-1}) = 
 f_{k-1} f_{k-2} + f_k(f_{k-2}-f_{k-1}) = 1$. This implies
  $f_{k-1}f_{k-2} \equiv 1 ~(\bmod f_k)$.
  Thus, $x$ is equal to $f_{k-1}-1$, and the
  the lexicographically smallest cyclic rotation of $F_k$
  starts at position $f_{k-1}$.

  This means that the second phrase of the lex-parse of $F_k$ starts
  one position before the lexicographically smallest cyclic rotation of $F_k$.
  So, now considering the terminator \$,
if a suffix $S$ of $F_k$ is lexicographically smaller than
  $F_{k}[f_{k-1}-1\dd ] = baF_{k-2}$ (i.e., the suffix that starts
  at the beginning of the second phrase of the lex-parse of $F_k$)
  and both share a common prefix $P$, then $S=P$ and $|S| < f_{k-2}+2$.
  Let us prove that $baF_{k-4}$ is the largest string that
  is a prefix and a suffix of $baF_{k-2}$.

  The string $F_{k-4}$ only occurs at positions $1, f_{k-4}+1,$ and
  $f_{k-3}+1$ within $F_{k-2}$:
  By Eq.~(\ref{eq2}), we have that
  $F_{k-2} = F_{k-4}F_{k-4}F_{k-7}F_{k-6}$.
  There are no occurrences of $F_{k-4}$ at positions 
$1 < p \le f_{k-4}$, by the same argument of Lemma~\ref{lem:noinside}.
  By Eq.~(\ref{eq3}), we also have that
  $F_{k-2} = F_{k-4}F_{k-5}F_{k-5}F_{k-6}$.
  There are no occurrences of $F_{k-4}$ at positions 
  $f_{k-4}+1 < p \le f_{k-3}$, because $F_{k-4} = F_{k-5}F_{k-6}$ and then
  $F_5$ would occur more than twice within $F_5F_5$, which is not possible
  again by the argument of Lemma~\ref{lem:noinside}.
  The last occurrence of $F_{k-4}$ within
  $F_{k-2} = F_{k-3} F_{k-4}$ must then be at position $f_{k-3}+1$.
  By Lemma~\ref{lem:near}, the only one of those three occurrences that is preceded
  by $ba$ is the last one.
  
  So the first two phrases of the lex-parse of $F_k$
  are of lengths $\ell_1 = f_{k-1}-2$, and $\ell_2 = f_{k-4}+2$,
  respectively.
  The rest $R_k$ of $F_k$ is then of length $f_{k-3}$.
  From Eq.~(\ref{eq2}), we have that
  $F_k = F_{k-2}F_{k-2}F_{k-5}F_{k-4}$, so
  $R_k = F_{k-5}F_{k-4} = H_{k-3}ab$, by Lemma~\ref{lem:near}.
Since $R_k$ starts with $H_{k-3}$, which starts with $F_{k-4}$ by 
Lemma~\ref{lem:near}, and it finishes with $F_{k-4}$, which is
the lexicographically smallest occurrence of $F_{k-4}$, we have
$\ell_3 \ge f_{k-4}$.

  By Lemma~\ref{lem:sub}, we have that all the
  suffixes of $F_k$ that start at position
  $1 \le p \le 2f_{n-2}$,
  and begin with $F_{k-4}$, also begin
  with $F_{k-4}F_{k-5} = H_{k-3}ba > R_k$, by Lemma~\ref{lem:near}, 
  or with $F_{k-5}F_{k-4} = R_k$.
  Since the suffix $R_k$ is followed by $\$$, 
  those suffixes are lexicographically larger than $R_k$.
  Also, $F_{k-4}$ occurs only at the beginning and at the end of $R_k = H_{k-3}ab$: $F_{k-4}$ only occurs at the beginning of $H_{k-3}$, by Lemmas~\ref{lem:near} and \ref{lem:noinside}, and because $R_k$ and $F_{k-4}$ both finish with $ab$,
  $F_{k-4}$ does not occur as a suffix of $H_{k-3}a$.
  So, the third phrase of the lex-parse of $F_k$ is of length
  $f_{k-4}$.
  
  The new rest $R'_k$ is of length $f_{k-5}$. Also, by Eq.~(\ref{eq2}),

  \begin{equation*}
    \begin{split}
      F_k & = F_{k-2}F_{k-2}F_{k-5}F_{k-4} \\
      & = F_{k-2}F_{k-2}F_{k-5}F_{k-5}F_{k-6} \\
      & = F_{k-2}F_{k-2}F_{k-5}F_{k-6}F_{k-7}F_{k-6} \\
      & = F_{k-2}F_{k-2}F_{k-4}F_{k-7}F_{k-6}.
    \end{split}
  \end{equation*}

  Then $R_{k-1} = F_{k-7}F_{k-6}$.
  Similarly as for $R_k$, by Lemma~\ref{lem:sub},
  all the occurrences of $F_{k-6}$ starting at
  positions $1\le p \le 2f_{k-2}+f_{k-4}$
  are lexicographically larger than $R_{k-1}$.
  Also, $F_{k-6}$ occurs only at the beginning and at the end of
  $F_{k-7}F_{k-6}$.
  We then have that the fourth phrase is of length
  $f_{k-6}$.

  The process continues in the same way up to $f_5$.
  At this point, the rest of $F_k$ is $aab$.
  We prove that the last three phrases of the lex-parse of $F_k$ are of length 1.
  First, the suffix $aab$ is the lexicographically smallest suffix
  of $F_k$ that begins with $a$, by Lemma~\ref{lem:bbaaa} and because $F_k$ is terminated in $\$$.
  Thus, the first $a$ of $aab$ is an explicit phrase of length 1.
  Then, the suffixes that are lexicographically smaller than $ab$
  begin with $aa$. Thus, the length of the next phrase is also 1.
  Finally, the suffix $b$ is the lexicographically smallest suffix of $F_k$ that begins with $b$.
  Thus, $b$ is an explicit phrase of length 1.
  
  Therefore, the lengths of the phrases of the lex-parse of $F_k$ are
  \[ f_{k-1}-1, f_{k-4}+2, f_{k-4}, f_{k-6}, \dots, f_5, 1, 1, 1 \]
  and the number of phrases is $5 + \frac{k-7}{2}$.
\end{proof}

\end{document}